\documentclass[a4paper,onecolumn,accepted=2019-10-20]{quantumarticle}
\pdfoutput=1
\usepackage[english]{babel}
\usepackage[T1]{fontenc}
\usepackage[utf8]{inputenc}
\usepackage{amsmath,amssymb}
\usepackage{amsthm}
\usepackage{indentfirst}
\usepackage{mathrsfs}
\usepackage{graphicx}
\usepackage{hyperref}
\usepackage{amsthm}
\usepackage{thmtools}
\usepackage{graphicx}
\usepackage{caption}
\usepackage{subcaption}
\usepackage[usenames, dvipsnames]{color}
\usepackage{csquotes}
\usepackage{verbatim}
\usepackage{bookmark}
\usepackage{booktabs}
\usepackage{url,hyperref}
\usepackage{enumerate}
\usepackage{dsfont}
\usepackage{geometry}
\usepackage[numbers,sort&compress]{natbib}

\providecommand{\keywords}[1]{\textbf{\textit{Keywords: }} #1}

\theoremstyle{plain}
\newtheorem{theorem}{Theorem}[]
\newtheorem{lemma}[theorem]{Lemma}

\newtheorem{coro}[theorem]{Corollary}

\theoremstyle{definition}

\newtheorem{defn}[theorem]{Definition}
\newtheorem{example}[theorem]{Example}

\newtheorem*{remark}{Remark}

\newcommand{\thmref}[1]{Theorem~\ref{#1}}
\newcommand{\corref}[1]{Corollary~\ref{#1}}

\newcommand{\lemref}[1]{Lemma~\ref{#1}}
\newcommand{\secref}[1]{\S~\ref{#1}}
\newcommand{\figref}[1]{Figure~\ref{#1}}

\newcommand{\ie}{{i.e.}}
\newcommand{\eg}{{e.g.}}

\makeatletter
\newcommand*{\rom}[1]{\expandafter\@slowromancap\romannumeral #1@}
\makeatother

\newcommand{\wt}[1]{\widetilde{#1}}

\DeclareMathOperator{\tr}{Tr}

\newcommand{\ud}{\,\mathrm{d}}

\newcommand{\MM}{\mathbb{M}}
\newcommand{\HH}{\mathbb{H}}

\newcommand{\Real}{\mathbb{R}}
\newcommand{\Complex}{\mathbb{C}}
\newcommand{\ee}{\mathbb{E}}

\newcommand{\vect}[1]{\textbf{#1}}

\newcommand{\eps}{\epsilon}

\newcommand{\Abs}[1]{\left\lvert#1\right\rvert}

\newcommand{\Norm}[1]{\left\lVert#1\right\rVert}

\newcommand{\Inner}[2]{\left\langle#1, #2\right\rangle}

\def\bigl{\mathopen\big}

\def\bigr{\mathclose\big}
\newcommand{\unit}{\mathbb{I}}
\newcommand{\id}{\mathcal{I}}

\newcommand{\Bra}[1]{\left\langle#1\right\rvert}

\newcommand{\Ket}[1]{\left\lvert#1\right\rangle}

\newcommand{\hbt}{\mathcal{H}}

\newcommand{\dset}{\mathfrak{D}}
\newcommand{\dsetplus}{\mathfrak{D}^{+}} 

\newcommand\myeq[1]{\mathrel{\stackrel{\makebox[0pt]{\mbox{\normalfont\tiny #1}}}{=}}}

\newcommand\myle[1]{\mathrel{\stackrel{\makebox[0pt]{\mbox{\normalfont\tiny #1}}}{\le}}}

\newcommand{\chisq}{\chi^2_{\kappa}}
\newcommand{\chisqarg}[2]{\chi^2_{\kappa}\left(#1 \mid\mid #2\right)}

\newcommand{\hs}{HS}
\newcommand{\half}{1/2}

\newcommand{\Omegaall}[1]{\Omega_{#1}}

\newcommand{\sdpiratio}{\mathsf{R}}
\newcommand{\Kop}{\mathscr{K}}
\newcommand{\Kset}{\mathcal{K}}

\usepackage{todonotes}

\newcommand{\revise}[1]{{#1}}
\newcommand{\revisetwo}[1]{{#1}}

\begin{document}
\title{Tensorization of the strong data processing inequality for quantum chi-square divergences}

\author{Yu Cao} 
\email{yucao@math.duke.edu} 
\affiliation{Department of Mathematics, Duke University, Box 90320, Durham NC 27708, USA}

\author{Jianfeng Lu} 
\email{jianfeng@math.duke.edu}
\affiliation{Department of Mathematics, Duke University, Box 90320, Durham NC 27708, USA}
\affiliation{Department of Physics and Department of Chemistry, Duke University, Box 90320, Durham NC 27708, USA}

\begin{abstract}
  \revise{It is well-known that any quantum channel $\mathcal{E}$
    satisfies the data processing inequality (DPI), with respect to
    various \revisetwo{divergences}, \eg, quantum $\chi^2_{\kappa}$
    divergences and quantum relative entropy.  More specifically, the data
    processing inequality states that the divergence between two
    arbitrary quantum states $\rho$ and $\sigma$ does not increase
    under the action of any quantum channel $\mathcal{E}$.  For a
    fixed channel $\mathcal{E}$ and a state $\sigma$,
    \revisetwo{the divergence between output states $\mathcal{E}(\rho)$ and $\mathcal{E}(\sigma)$ might be strictly smaller than the divergence between input states $\rho$ and $\sigma$, which is
    characterized by the strong data processing inequality (SDPI).
 Among various input states $\rho$, the largest value of the rate of contraction is known as the SDPI constant.}   
    An important and widely studied property for classical channels is
    that SDPI constants tensorize. In this paper, we extend the
    tensorization property to the quantum regime: we establish the
    tensorization of SDPIs for the quantum $\chi^2_{\kappa_{1/2}}$ divergence
    for arbitrary quantum channels and also for a family of $\chisq$
    divergences (with $\kappa \ge \kappa_{1/2}$) for arbitrary quantum-classical channels.}
\end{abstract}

\maketitle

\section{Introduction} 

\revise{
In information theory, the data processing inequality (DPI) has been an important property for \revisetwo{divergence measures} to possess operational meaning. For instance, DPI has been proved for quantum $\chisq$ divergences 
(see \eg{}, 
\cite[Thm. II.14]{lesniewski_monotone_1999} or
\cite[Thm. 4]{temme_divergence_2010}), among other divergences. 
More explicitly, for any quantum channel
$\mathcal{E}$ and for all quantum states $\rho, \sigma\in \dset_n$, we have
\begin{equation}
\label{eqn::spi}
\chisqarg{\mathcal{E}(\rho)}{\mathcal{E}(\sigma)} \le \chisqarg{\rho}{\sigma}.
\end{equation}
In the above, $\kappa$ is a real-valued positive function (see
\eqref{eqn::set_K} below); the definition of $\chisq$ divergences will
be postponed to \secref{subsec::quantum_chisq_divergences}, as it
involves some technicalities.

Compared with the DPI, the strong data processing inequality (SDPI) quantitatively and more
precisely characterizes the extent that quantum states contract under
the channel $\mathcal{E}$
\cite{anantharam_maximal_2013, Raginsky_strong_2016, polyanskiy_dissipation_2016,
  polyanskiy_strong_2017}.  Given any $(\mathcal{E}, \sigma)$-pair
where $\mathcal{E}$ is any quantum channel and $\sigma\in \dsetplus_n$
is any full-rank quantum state ($\dsetplus_n$ is the space of strictly
positive density matrices on a $n$-dimensional Hilbert space), if
there is a constant
$\eta_{\chisq}\left(\mathcal{E}, \sigma \right) \in [0, 1)$ such that
\begin{equation}
\label{eqn::sdpi}
\chisqarg{\mathcal{E}(\rho)}{\mathcal{E}(\sigma)} \le \eta_{\chisq}\left(\mathcal{E}, \sigma \right)  \chisqarg{\rho}{\sigma},\ \forall \rho \in \dset_n,
\end{equation}
then the quantum channel $\mathcal{E}$ is said to satisfy the
\emph{strong data processing inequality} (SDPI) for the quantum
$\chisq$ divergence and the smallest constant
$\eta_{\chisq}\left(\mathcal{E}, \sigma \right)$ such that
\eqref{eqn::sdpi} holds
is called the \emph{SDPI constant}.
Evidently,
\begin{equation}
\label{eqn::sdpi_const}
\eta_{\chisq}\left(\mathcal{E}, \sigma \right) = \sup_{\rho\in \dset_n:\ \rho \neq \sigma} \frac{\chisqarg{\mathcal{E}(\rho)}{\mathcal{E}(\sigma)}}{\chisqarg{\rho}{\sigma}}.
\end{equation}
Many applications of SDPIs can be found in
\eg{}, \cite[Sec. 2.3]{polyanskiy_strong_2017} and
\cite[Sec. \rom{5}]{Raginsky_strong_2016}.

It is common in quantum information theory to consider
high-dimensional quantum channels, formed by the tensor product of
low-dimensional quantum channels.  Except for very special cases, in
general, obtaining SDPI constants for high-dimensional quantum channels can be
rather challenging, even numerically. 
It is desirable if one could reduce the
problem of calculating the SDPI constant for a (global)
high-dimensional quantum channel, to calculating the SDPI constants of low-dimensional quantum channels. 
For a specific \revisetwo{divergence} (\eg{}, quantum $\chisq$ divergence in this work), if the SDPI constant for the high-dimensional channel is the maximum value of SDPI constants for these low-dimensional channels, we say that the SDPI constant for this \revisetwo{divergence} satisfies the \emph{tensorization property}.

Our main result in this work is that the SDPI constant for $\chisq$
tensorizes,
summarized in the following theorem. 
\begin{theorem}
\label{thm::chisq_tensorize}
Consider $N$ finite-dimensional quantum systems whose Hilbert spaces are $\hbt_j$ with dimension $n_j$ ($1\le j \le N$) 
and consider any density matrix $\sigma_j\in \dsetplus_{n_j}$ and any quantum channel $\mathcal{E}_j$ acting on $\hbt_j$, 
such that for all $1\le j \le N$, $\mathcal{E}_j (\sigma_j)\in \dsetplus_{n_j}$.
If either of the followings holds 
\begin{enumerate}[(i)]
\item $\kappa = \kappa_{1/2}$; 
\item $\kappa \ge \kappa_{1/2}$ and $\mathcal{E}_j$ are quantum-classical (QC) channels;
\end{enumerate}
then we have the tensorization of the SDPI constant for the quantum $\chisq$ divergence, \ie, 
\begin{equation}\label{eq:mainthm}
\eta_{\chisq}\left(\mathcal{E}_1\otimes \mathcal{E}_2 \otimes \cdots \otimes \mathcal{E}_{N}, \sigma_1\otimes \sigma_2 \otimes \cdots \otimes \sigma_{N}\right) = \max_{1\le j \le N} \eta_{\chisq}\left(\mathcal{E}_j, \sigma_j \right).
\end{equation}
\end{theorem}

\begin{remark}
\begin{enumerate}[(i)]
\item The function $\kappa_{1/2}(x) := x^{-1/2}$ is a special example of weight functions. 
There are some properties that only the quantum $\chi_{\kappa_{1/2}}^2$ divergence possesses (see \eg{}, \lemref{lem::prop_Omega_op_composite} (\ref{lem::prop_Omega_op_composite_tensorize})); 
in addition, $\chi_{\kappa_{1/2}}^2$ is tightly connected to the sandwiched R{\'e}nyi divergence of order $2$ \cite{muller-lennert_quantum_2013}.  

 There is a whole family of $\kappa_{\alpha}$ parameterized by
$\alpha\in [0,1]$, satisfying the condition
$\kappa_{\alpha} \ge \kappa_{1/2}$; see Example \ref{eg::chi_alpha} for
details; in \secref{subsec::example_kappa}, we also present other
examples of $\kappa(x)$ such that $\kappa \ge \kappa_{1/2}$. The notion
of QC channel will be recalled in \secref{subsec::comparison}.

\item  
These assumptions only provide sufficient conditions for the tensorization of SDPIs to hold, and it is an interesting open question to further investigate weaker conditions. 
In addition, it is also an interesting open question whether the tensorization of SDPIs holds for (quasi) relative entropies and the geodesic distances \cite{lesniewski_monotone_1999, hiai_contraction_2016}. 
We shall leave these questions to future research.

\end{enumerate}
\end{remark}

The tensorization property in the classical regime has been well studied and widely used; 
see \eg, \cite{anantharam_maximal_2013,
van2016probability,
Raginsky_strong_2016}.  
For SDPI constants, the tensorization property was proved in
\cite[Thm. \rom{3}.9]{Raginsky_strong_2016} 
for any $\Phi$-divergence, denoted by
$
  D_{\Phi}\left(\nu \mid\mid \mu \right) := \ee_{\mu}\Big[\Phi(\frac{\ud \nu}{\ud \mu})\Big] - \Phi(1),
$
provided that the associated $\Phi$-entropy is sub-additive and homogeneous. 
As a remark, the $\Phi$-divergence includes the relative entropy (with $\Phi(x) = x\log(x)$) and the classical $\chi^2$ divergence (with $\Phi(x) = (x-1)^2$) as special instances. 
The tensorization of SDPI constants associated with the classical relative entropy has been applied to study the lower bounds of Bayes risk \cite{xu_converses_2015}.

Establishing tensorization in the quantum regime seems to be more
challenging and our understanding is much limited.
Recently, the tensorization technique has been developed for the
quantum hypercontractivity of qubit system \cite{King_2014}, reversed
hypercontractivity \cite{cubitt_quantum_2015,beigi_quantum_2018},
$2$-log-Sobolev constant \cite{Kastoryano_quantum_2013,
  beigi_quantum_2018}, as well as the quantum maximal correlation
\cite{beigi_new_2013}. For the tensorization of the quantum (reversed)
hypercontractivity and log-Sobolev constants, all existing works, as
far as we know, focus exclusively on reversible (or even more special)
quantum Markov semigroups (\ie{}, Lindblad equations).

We would like to briefly mention and highlight the proof techniques used for \thmref{thm::chisq_tensorize}. The first main ingredient is to formulate the SDPI constant as the second largest eigenvalue of a certain operator (see \lemref{lem::sdpi_variational}); similar results have been obtained in \eg{}, \cite{choi_equivalence_1994, Raginsky_strong_2016,ruskai_beyond_1994, lesniewski_monotone_1999, hiai_contraction_2016}.
This result immediately leads into the proof of the case (i). The second main ingredient is to bound $\eta_{\chi^2_{\kappa_{1/2}}}(\mathcal{E}, \sigma)$ above by $\sqrt{\eta_{\chisq}(\mathcal{E}, \sigma)}$ (see \lemref{lem::chialphahalf_less_chihalfsquare}), whose proof uses Petz recovery map \cite{petz_sufficiency_1988} as the bridge. This relation together with special properties of $\eta_{\chi^2_{\kappa_{1/2}}}(\mathcal{E}, \sigma)$ leads into the proof of case (ii).

\subsection*{Related techniques to quantify the loss of information}
Apart from the DPI and the SDPI, there are other concepts used to characterize the contraction of quantum states under the action of noisy channels. 
For instance, one widely studied quantity is 
the \emph{contraction coefficient} 
\begin{equation}
\label{eqn::contraction_coef}
\eta_{\chisq}(\mathcal{E}) := \sup_{\sigma \in \dsetplus_n}\ \eta_{\chisq}(\mathcal{E},\sigma).
\end{equation} 
The contraction coefficient is very similar to the SDPI constant. However, compared with the SDPI constant, the contraction coefficient for various divergence measures has been much more extensively
studied in the literature, see \eg{}, \cite{ruskai_beyond_1994,petz_contraction_1998,lesniewski_monotone_1999,temme_divergence_2010,
  hiai_contraction_2016} for the quantum case, and see \eg{}, \cite{cohen_relative_1993} for the classical case. The bijection maps
that preserve the quantum $\chi^2_{\kappa_{\alpha}}$ divergence (see
Example \ref{eg::chi_alpha} about the family $\kappa_{\alpha}$) have
been characterized in \cite{chen_maps_2017}, which complements the
study of the contraction of quantum states.
There are other tools based on the functional perspective, including quantum (reverse) hypercontractivity and related quantum functional
inequalities \cite{Kastoryano_quantum_2013, King_2014,
cubitt_quantum_2015,
rouze_concentration_2017, 
beigi_quantum_2018}.

\subsection*{Contribution} 
We summarize new results obtained in this work, as follows:
\begin{enumerate}[(i)]
\item Our main result is \thmref{thm::chisq_tensorize}, which
  establishes the tensorization of SDPI constants, under certain
  assumptions: for the quantum $\chi^2_{\kappa_{1/2}}$ divergence, the
  tensorization of SDPI constants holds for general quantum channels; for
  the quantum $\chi_{\kappa}^2$ divergence with $\kappa \ge \kappa_{1/2}$,
  the tensorization holds for any quantum-classical channel.

\item Along the analysis of the SDPI, we also establish a connection
  between the SDPI constant associated with $\kappa_{1/2}$ and
  a variant of quantum maximal correlations; see
  \thmref{thm::sdpi_max_corr} for details.

\item To use the tensorization property, we need to understand
  the SDPI constants for local channels, 
{\ie, we need to compute $\eta_{\chisq}(\mathcal{E}_j, \sigma_j)$ for $1\le j\le N$.} 
Motivated by this, we study the SDPI constants for special {qubit channels} in \secref{sec::example}. 
We notice that there is a particular QC channel $\mathcal{E}$ associated with a fixed $\sigma\in \dsetplus_2$ such that the largest value of
$\eta_{\chisq}(\mathcal{E}, \sigma) \approx 1$ for
$\kappa = \kappa_{\min}$, while
$\eta_{\chisq}(\mathcal{E}, \sigma) \approx 0$ for
$\kappa = \kappa_{\max}$ (however, $\sigma$ is close to a singular matrix); 
see \secref{subsec::qc_channel} for details. 
This extreme example shows the high dependence of SDPI constants on the choice of $\kappa$, 
which magnifies the difference between the quantum SDPI constant and its classical analog, because there is only one SDPI constant for the classical $\chi^2$ divergence.

\end{enumerate}
}

This paper is organized as follows. In \secref{sec::preliminary}, we provide some preliminary results, in particular, we recall the eigenvalue formalism of the SDPI constant. In \secref{sec::chisq_tensorize}, we prove \thmref{thm::chisq_tensorize} and in \secref{sec::max_corr}, we study the connection between the SDPI constant and the quantum maximal correlation. 
In \secref{sec::example},  we consider SDPI constants for {qubit channels}
and study the dependence of $\eta_{\chisq}(\mathcal{E}, \sigma)$ on $\sigma$ and $\kappa$. \secref{sec::discussion} concludes the paper with some additional remarks.

\section{Preliminaries}
\label{sec::preliminary}
This section contains preliminary results that we will use to prove
the tensorization of the strong data processing inequality,
\thmref{thm::chisq_tensorize}. In particular, we will present two
variational formulations of SDPI constants, and discuss the relation
between various SDPI constants.

\medskip

{\bf \noindent Notations.} We shall consider finite dimensional systems only,  \ie, the Hilbert space $\hbt \cong \Complex^n$. 
Let $\MM_n$, $\dset_n$, $\dsetplus_n$, $\HH_n$ be the 
space of linear operators, density matrices, strictly positive density matrices and Hermitian matrices on $\hbt$, respectively. Let $\MM_n^0$ and $\HH_n^0$ be the space of traceless elements of $\MM_n$ and $\HH_n$, respectively.
Denote the $n$-by-$n$ identity matrix by $\unit_n$ (acting on $\hbt$); 
let $\id_{n}$ be the identity operator acting on $\MM_n$. 
If the Hilbert space $\hbt = \hbt_1\otimes \hbt_2\otimes \cdots \otimes \hbt_N$, and $\hbt_j$ has the dimension $n_j$ (for $1\le j\le N$), then the space of linear operators on $\hbt$ is denoted by $\MM_{n_1\times n_2 \times \cdots \times n_N}$; the same convention applies similarly to other spaces, \eg{}, $\HH_{n_1\times n_2 \times \cdots \times n_N}$. As a reminder, following the above notation convention, $\HH^0_{n_1\times n_2} \neq \HH^0_{n_1}\otimes \HH^0_{n_2}$.

Let $\Inner{\cdot}{\cdot}$ denote a generic inner product on $\MM_n$; the Hilbert-Schmidt inner product is defined as $\Inner{A}{B}_{\hs} := \tr\left(A^{\dagger} B\right)$. 
\revise{For any positive semidefinite operator $\mathscr{T}$ on $\MM_n$,
define the sesquilinear form $\Inner{A}{B}_{\mathscr{T}} := \Inner{A}{\mathscr{T}(B)}_{\hs}$ and the semi-norm $\Norm{A}_{\mathscr{T}} := \sqrt{\Inner{A}{A}_{\mathscr{T}}}$ for all $A,\ B\in \MM_n$;
when $\mathscr{T}$ is strictly positive, 
the sesquilinear form becomes an inner product and the semi-norm becomes a norm. 
}

For convenience, for any $A, B\in \MM_n$, we denote
\begin{align}
A\#B := L_{A} R_{B},  
\end{align}
where $L_A$ and $R_B$ are left and right multiplication of $A$ and $B$, respectively; in other words,  $(A\#B)(X) = AXB$.  

\subsection{Quantum $\chisq$ divergences}
\label{subsec::quantum_chisq_divergences}

Throughout this work, we consider the quantum $\chi_{\kappa}^2$ divergence, introduced in \cite[Def. 1]{temme_divergence_2010}.  
Let us introduce a set $\mathcal{K}$,
\begin{align}
\label{eqn::set_K}
\Kset := \left\{\kappa:(0,\infty)\rightarrow (0,\infty)| -\kappa \text{ is operator monotone}, \kappa(1) = 1, x \kappa(x) = \kappa(x^{-1})\right\}.
\end{align} 
As a remark, it is easy to check that
$\kappa_{1/2}(x) \equiv x^{-1/2}$ is in the family $\mathcal{K}$.

\begin{defn}[Quantum $\chi^2_{\kappa}$ divergence]
For any $\kappa \in \mathcal{K}$, define the quantum $\chi^2_{\kappa}$ divergence between quantum states $\rho, \sigma\in \dset_n$ by  
\begin{align}
\label{eqn::chisq}
\chisqarg{\rho}{\sigma} &:= \Inner{\rho-\sigma}{  \Omega_{\sigma}^{\kappa}(\rho-\sigma)}_{\hs}
\end{align}
when $\text{supp}(\rho)\subset \text{supp}(\sigma)$; otherwise, set $\chisqarg{\rho}{\sigma} = \infty$.
The operator $\Omega_{\sigma}^{\kappa}$ above is given by 
\begin{align}
\Omega_{\sigma}^{\kappa} := R_{\sigma}^{-1} \kappa(L_{\sigma} R_{\sigma}^{-1}) \equiv L_{\sigma}^{-1} \kappa(R_{\sigma} L_{\sigma}^{-1}).
\end{align}
The second equality comes from the assumption that $x\kappa(x) = \kappa(x^{-1})$. 
\revise{As a remark, when $\sigma$ is not a full-rank density matrix, $\Omega_{\sigma}^{\kappa}$ can still be well-defined on the support of $\sigma$.}
\end{defn}

Essentially, the operator
$\Omega_{\sigma}^{\kappa}$ is a non-commutative way to multiply
$\sigma^{-1}$. Properties of the operator $\Omega_{\sigma}^{\kappa}$ will be further discussed 
in \secref{subsec::basic_properties}.

Next, let us introduce the non-commutative way to multiply $\sigma$. Define the weight operator $\Gamma_{\sigma}$ 
\begin{align}
\label{eqn::op_gamma}
\Gamma_{\sigma} := \sigma^{\half} \# \sigma^{\half}. 
\end{align}
Note that the operator $\Gamma_{\sigma}$ is completely positive, with the Kraus operator $\sigma^{\half}$ 
and $\Omega_{\sigma}^{\kappa_{1/2}} = (\Gamma_{\sigma})^{-1}$.
For any $\kappa \in \mathcal{K}$, let us define a generalization of the operator $\Gamma_{\sigma}$
\begin{align}
\mho_{\sigma}^{\kappa} :=& L_{\sigma} \kappa(L_{\sigma} R_{\sigma}^{-1}) \equiv L_{\sigma} \kappa(\sigma\# \sigma^{-1}) \label{eqn::mho_op} \\
=& \Gamma_{\sigma} \circ \Omega_{\sigma}^{\kappa} \circ \Gamma_{\sigma}. \label{eqn::mho_op_kappa}
\end{align}
Notice that $\mho_{\sigma}^{\kappa_{1/2}} = \Gamma_{\sigma}$.

\subsection{Examples of $\kappa(x)$}
\label{subsec::example_kappa}

In this subsection, we provide three examples of $\kappa$ such that $\kappa \ge \kappa_{1/2}$ (satisfying one of the conditions in \thmref{thm::chisq_tensorize}). More examples can be found in \cite[Sec. 4.2]{hiai_families_2013} and \cite[Sec. (III)]{hiai_contraction_2016}.
\begin{example}[Quantum $\chi^2_{\kappa_{\alpha}}$ divergence]
\label{eg::chi_alpha}
An important family of the quantum $\chi^2_{\kappa}$ divergence is the quantum $\chi^2_{\kappa_{\alpha}}$ divergence, with the parameter $\alpha\in[0,1]$ and 
\begin{align}
\label{eqn::kappa-chi-alpha}
\kappa_{\alpha}(x) = \frac{1}{2}\left(x^{-\alpha}+x^{\alpha-1}\right).
\end{align}
\begin{enumerate}[(i)]
\item The case $\alpha=\half$ is very special: $\kappa_{1/2}(x) = x^{-1/2}$ and \revise{$\Omega_{\sigma}^{\kappa_{1/2}} = {\sigma^{-\half}} \# {\sigma^{-\half}}$} is completely positive with the Kraus operator $\sigma^{-\half}$. In fact, $\kappa_{\half}$ is the only one in $\Kset$ such that for any $\sigma$, both $\Omega_{\sigma}^{\kappa}$ and $\left(\Omega_{\sigma}^{\kappa}\right)^{-1}$ are completely positive \cite[Theorem 3.5]{hiai_families_2013}.

\item We can immediately verify that $\kappa_{\alpha} = \kappa_{1-\alpha}$ and for any fixed $x\in (0,\infty)$, 
$\kappa_{\alpha}(x)$ is monotonically decreasing with respect to $\alpha\in[0,\half]$; thus $\kappa_{\alpha}(x) \ge \kappa_{\half}(x)$. 
\end{enumerate}
\end{example}
More results about this family of the quantum $\chi^2_{\kappa_{\alpha}}$ divergence (also called \emph{mean $\alpha$-divergence}) could be found in \cite{temme_divergence_2010}.

\begin{example}[Wigner-Yanase-Dyson]
Another family of $\kappa^{WYD}_{\beta}$ \revise{(see \eg{}, \cite{hiai_contraction_2016,hiai_families_2013})}
corresponds to the Wigner-Yanase-Dyson metric, and it is parameterized by $\beta\in [-1,2]$, 
\begin{equation}
\label{eqn::kappa_WYD}
\kappa^{WYD}_{\beta}(x) := 
 \frac{1}{\beta(1-\beta)} \frac{(1-x^{\beta})(1-x^{1-\beta})}{(1-x)^2},\qquad x\in (0,1)\cup (1,\infty).
\end{equation}
When $x = 1$, \revise{${\kappa}^{WYD}_{\beta}(x)$ is simply set as $1$ or is defined by taking the limit $x\rightarrow 1$ in the above equation.}
In general, finding all possible $\beta\in [-1,2]$ such that $\kappa^{WYD}_{\beta} \ge \kappa_{\half}$ seems to be slightly technical; 
however, at least, for a few special choices of $\beta$, 
\eg, when $\beta = 1.5$ ($\kappa^{WYD}_{1.5}(x) \equiv \kappa_{\half}(x) + \frac{x^{-1/2} (\sqrt{x}-1)^4)}{3(1-x)^2}$) 
 and $\beta = 2$ ($\kappa^{WYD}_2(x) \equiv \frac{1+x}{2x}$) we can easily check that ${\kappa}^{WYD}_{\beta} \ge \kappa_{\half}$ for these two cases.

\end{example}

\begin{example}[The largest possible $\kappa$]
	The largest $\kappa\in \mathcal{K}$ is $\kappa_{\max} := \frac{1+x}{2x}$ \revise{(see \eg{}, \cite[Eq. (11)]{hiai_contraction_2016})}.
It is obvious that $\kappa_{\max} \ge \kappa_{\half}$. 
\end{example}
As a remark, ${\kappa}^{WYD}_2$ in the family of  Wigner-Yanase-Dyson metric is exactly the maximum one.

\subsection{Basic properties of operators $\Omega_{\sigma}^{\kappa}$ and $\mho_{\sigma}^{\kappa}$}
\label{subsec::basic_properties}
We list without proof some elementary while useful properties of the operator $\Omega_{\sigma}^{\kappa}$.
Recall the assumption that $x\kappa(x) = \kappa(x^{-1})$, which is used below in the proof of $\Omega_{\sigma}^{\kappa}$ being Hermitian-preserving.
\begin{lemma}
\label{lem::prop_Omega_op}
Suppose $\sigma\in \dsetplus_n$ and its eigenvalue decomposition $\sigma = \sum_{j=1}^{n} s_j \Ket{s_j}\Bra{s_j}$. Then 
\begin{enumerate}[(i)]

\item The operator $\Omega_{\sigma}^{\kappa}$ can be decomposed as
\begin{align}
\label{eqn::Omega_decomp}
\Omega_{\sigma}^{\kappa} = \sum_{j,m=1}^{n} \kappa\left(\frac{s_j}{s_m}\right) \frac{1}{s_m}\ {\Ket{s_j}\Bra{s_j}} \# {\Ket{s_m}\Bra{s_m}}. 
\end{align} 
For any Hermitian matrix $A\in \MM_n$,  
\begin{align}
\label{eqn::AA_omega_kappa}
\Inner{A}{A}_{\Omega_{\sigma}^{\kappa}} \equiv \Inner{A}{\Omega_{\sigma}^{\kappa} (A)}_{\hs} = \sum_{j,m=1}^{n} \kappa\left(\frac{s_j}{s_m}\right) \frac{1}{s_m} \Abs{\Bra{s_j} A \Ket{s_m}}^2 \ge 0.
\end{align}
Thus, $\Omega_{\sigma}^{\kappa}$ is a strictly positive operator with respect to the Hilbert-Schmidt inner product, and the inner product $\Inner{\cdot}{\cdot}_{\Omega_{\sigma}^{\kappa}}$ is well-defined. \label{lem::prop_Omega_op_positivity}

\item $\Omega_{\sigma}^{\kappa}$ is Hermitian-preserving. \label{lem::prop_Omega_op_hermitian_preserving}

\item We have $\Omega_{\sigma}^{\kappa}(\sigma) = \unit_n$. Thus
for any $A\in \MM_n$, 
\begin{align}
\label{eqn::inner_omega_unit}
\Inner{A}{\sigma}_{\Omega_{\sigma}^{\kappa}} = \Inner{A}{\unit_n}_{\hs}, \qquad \Inner{\sigma}{A}_{\Omega_{\sigma}^{\kappa}} = \Inner{\unit_n}{A}_{\hs} = \tr(A).
\end{align}
In particular, for any density matrix $\rho\in \dset_n$, 
$\Inner{\rho}{\sigma}_{\Omega_{\sigma}^\kappa} = \Inner{\sigma}{\rho}_{\Omega_{\sigma}^{\kappa}} = 1$.
\label{lem::prop_Omega_op_unity}
\end{enumerate}
\end{lemma}

Then let us consider the properties of $\Omega_{\sigma}^{\kappa}$ for a composite system. 
\begin{lemma}
\label{lem::prop_Omega_op_composite}

\begin{enumerate}[(i)]
\item Consider $\sigma_1\in \dsetplus_{n_1}$ and $\sigma_2\in \dsetplus_{n_2}$. Then for any $A\in \MM_{n_1}$ and $B\in \MM_{n_2}$, we have
\begin{align}
\label{eqn::Omega_op_composite}
\Omega_{\sigma_1\otimes \sigma_2}^{\kappa}(A \otimes \sigma_2) = \Omega_{\sigma_1}^{\kappa} (A) \otimes \unit_{n_2} \qquad \Omega_{\sigma_1\otimes \sigma_2}^{\kappa}(\sigma_1 \otimes B) =\unit_{n_1} \otimes  \Omega_{\sigma_2}^{\kappa} (B). 
\end{align}

\item \label{lem::prop_Omega_op_composite_tensorize}
$\kappa_{\half}$ is the only one in $\Kset$ such that 
for all $\sigma_1\in \dsetplus_{n_1}$ and $\sigma_2 \in \dsetplus_{n_2}$, we have
\begin{align}
\label{eqn::omega_op_tensorize}
\Omega_{\sigma_1\otimes \sigma_2}^{\kappa} = \Omega_{\sigma_1}^{\kappa} \otimes \Omega_{\sigma_2}^{\kappa}.
\end{align}
\end{enumerate}
\end{lemma}

\begin{proof}
Let us decompose $\sigma_1 = \sum_{j=1}^{n_1} \lambda_j \Ket{\psi_j}\Bra{\psi_j}$ and $\sigma_2 = \sum_{m=1}^{n_2} \mu_{m} \Ket{\phi_m}\Bra{\phi_m}$, then $\sigma_1 \otimes \sigma_2$ has an eigenvalue decomposition $\sigma_1 \otimes \sigma_2 = \sum_{j, m} \lambda_j \mu_m \Ket{\psi_j}\Bra{\psi_j} \otimes \Ket{\phi_m}\Bra{\phi_m}$. 
\begin{enumerate}[(i)]
\item By the decomposition of the operator $\Omega^{k}_{(\cdot)}$ in \eqref{eqn::Omega_decomp}, 
\begin{align*}
\Omega_{\sigma_1\otimes \sigma_2}^{\kappa} = \sum_{j_1, j_2, m_1, m_2} \kappa\left(\frac{\lambda_{j_1}\mu_{m_1}}{\lambda_{j_2} \mu_{m_2}}\right) \frac{1}{\lambda_{j_2} \mu_{m_2}} \left(\Ket{\psi_{j_1}}\Bra{\psi_{j_1}} \otimes \Ket{\phi_{m_1}} \Bra{\phi_{m_1}}\right) \# \left( \Ket{\psi_{j_2}}\Bra{\psi_{j_2}} \otimes \Ket{\phi_{m_2}} \Bra{\phi_{m_2}}\right).
\end{align*}
Then by direct calculation, 
\begin{align*}
& \Omega_{\sigma_1\otimes \sigma_2}^{\kappa}(A\otimes \sigma_2) \\
=& \sum_{j_1, j_2, m_1, m_2} \kappa\left(\frac{\lambda_{j_1}\mu_{m_1}}{\lambda_{j_2} \mu_{m_2}}\right) \frac{1}{\lambda_{j_2} \mu_{m_2}} 
\mu_{m_1} \delta_{m_1, m_2} \left(\Ket{\psi_{j_1}}\Bra{\psi_{j_1}} A \Ket{\psi_{j_2}}\Bra{\psi_{j_2}}\right)\otimes \Ket{\phi_{m_1}}\Bra{\phi_{m_1}} \\
=& \sum_{j_1, j_2, m_1} \kappa\left(\frac{\lambda_{j_1}}{\lambda_{j_2} }\right) \frac{1}{\lambda_{j_2}} 
 \left(\Ket{\psi_{j_1}}\Bra{\psi_{j_1}} A \Ket{\psi_{j_2}}\Bra{\psi_{j_2}}\right)\otimes \Ket{\phi_{m_1}} \Bra{\phi_{m_1}}\\
 =&\ \Omega_{\sigma_1}^{\kappa}(A) \otimes \unit_{n_2}. 
\end{align*}
The other case can be  similarly proved. 

\item When $\kappa = \kappa_{\half}$, by the fact that $\Omega_{\sigma}^{\kappa_{1/2}} = (\Gamma_{\sigma})^{-1}$, we can immediately see the tensorization \eqref{eqn::omega_op_tensorize}. As for the other direction, from the assumption that \eqref{eqn::omega_op_tensorize} holds and after some straightforward simplification, one could obtain that 
$\kappa\left(\frac{\lambda_{j_1} \mu_{m_1}}{\lambda_{j_2} \mu_{m_2}}\right) = \kappa\left(\frac{\lambda_{j_1}}{\lambda_{j_2}}\right) \kappa\left(\frac{\mu_{m_1}}{\mu_{m_2}}\right)$, for all indices $j_1, j_2, m_1, m_2$. Since $\sigma_1$ and $\sigma_2$ are arbitrary density matrices, we have $\kappa(xy) = \kappa(x)\kappa(y)$ for all $x, y> 0$; in particular,  $1 = \kappa(1) = \kappa(x) \kappa(x^{-1})$. Since $\kappa\in \Kset$, we also have $x\kappa(x) = \kappa(x^{-1})$, which 
leads into $\kappa(x) = x^{-1/2} = \kappa_{\half}$.
\end{enumerate}
\end{proof}

Similarly, we list without proof the following properties of $\mho_{\sigma}^{\kappa}$; all properties can be easily verified by the definition of $\mho_{\sigma}^{\kappa}$ in \eqref{eqn::mho_op}.  
\begin{lemma}[Operator $\mho_{\sigma}^{\kappa}$]
\label{lem::mho_op}
Suppose $\sigma\in \dsetplus_n$ and its eigenvalue decomposition $\sigma = \sum_{j=1}^{n} s_j \Ket{s_j}\Bra{s_j}$. Then 
\begin{enumerate}[(i)]
\item the operator $\mho_{\sigma}^{\kappa}$ for any $\kappa\in \mathcal{K}$ has a decomposition
\begin{align}
\label{eqn::mho_decomp}
\mho_{\sigma}^{\kappa} = \sum_{j,m} s_j \kappa\left(\frac{s_j}{s_m}\right)\ {\Ket{s_j}\Bra{s_j}} \# {\Ket{s_m}\Bra{s_m}},
\end{align}
thus $\mho_{\sigma}^{\kappa}$ is strictly positive with respect to the Hilbert-Schmidt inner product;
\item the operator $\mho_{\sigma}^{\kappa}$ is Hermitian-preserving;
\item $\mho_{\sigma}^{\kappa} (\unit_n) = \sigma$.
\end{enumerate}
\end{lemma}

\subsection{Eigenvalue formalism of SDPI constants}
The eigenvalue formalism of the quantum contraction coefficient can be found in \eg{} \cite{ruskai_beyond_1994, lesniewski_monotone_1999, hiai_contraction_2016}; 
the  classical analogous result can be found in \eg, \cite{choi_equivalence_1994, Raginsky_strong_2016}. In this subsection, we concisely present this formalism, for the sake of completeness.

Let us consider the ratio in the SDPI constant.
\begin{align}
\label{eqn::sdpi_ratio}
\frac{\chisqarg{\mathcal{E}(\rho)}{\mathcal{E}(\sigma)}}{\chisqarg{\rho}{ \sigma}} &= \frac{\Inner{\rho-\sigma}{\mathcal{E}^{\dagger}\circ \Omega_{\mathcal{E}(\sigma)}^{\kappa} \circ  \mathcal{E}(\rho-\sigma)}_{\hs} 
}{\Inner{\rho-\sigma}{\Omega_{\sigma}^{\kappa} (\rho-\sigma)}_{\hs}} = \frac{\Inner{\rho-\sigma}{\Upsilon_{\mathcal{E},\sigma}^{\kappa} (\rho-\sigma)}_{\Omega_{\sigma}^{\kappa}} 
}{\Inner{\rho-\sigma}{\rho-\sigma}_{\Omega_{\sigma}^{\kappa}}},
\end{align}
where we introduce 
\begin{equation}
\label{eqn::op_Upsilon}
\begin{split}
\Upsilon_{\mathcal{E},\sigma}^{\kappa} &:= 
\left(\Omega_{\sigma}^{\kappa}\right)^{-1} \circ \mathcal{E}^\dagger\circ \Omegaall{\mathcal{E}(\sigma)}^{\kappa} \circ \mathcal{E}.
\end{split}
\end{equation}

Here are some properties of the operator $\Upsilon_{\mathcal{E}, \sigma}^{\kappa}$.
\begin{lemma}
\label{lem::Upsilon_prop}
\revise{Assume that $\sigma, \mathcal{E}(\sigma)\in \dsetplus_n$.}
\begin{enumerate}[(i)]
\item  \label{lem::Upsilon_prop_positive}
 The operator $\Upsilon_{\mathcal{E},\sigma}^{\kappa}$ is positive semidefinite with respect to the inner product $\Inner{\cdot}{\cdot}_{\Omega_{\sigma}^{\kappa}}$.
 
\item \label{lem::Upsilon_prop_fixed_point} $\Upsilon_{\mathcal{E},\sigma}^{\kappa}(\sigma) = \sigma$.

\item \label{lem::Upsilon_prop_hermitian}
$\Upsilon_{\mathcal{E},\sigma}^{\kappa}$ is Hermitian perserving.

\item \label{lem::Upsilon_prop_bound_by_one} For any $A\in \MM_n$, we have
\begin{align*}
\Inner{A}{\Upsilon_{\mathcal{E},\sigma}^{\kappa}(A)}_{\Omega_{\sigma}^{\kappa}} \le \Inner{A}{A}_{\Omega_{\sigma}^{\kappa}}. 
\end{align*}
Therefore, the eigenvalue of $\Upsilon_{\mathcal{E},\sigma}^{\kappa}$ is bounded above by $1$.
\end{enumerate}
\end{lemma}

\begin{proof}
Part (i) is obvious from \eqref{eqn::op_Upsilon} and \lemref{lem::prop_Omega_op} (\ref{lem::prop_Omega_op_positivity}). 
Part (ii) can be verified directly by \lemref{lem::prop_Omega_op} (\ref{lem::prop_Omega_op_unity}) \revise{and the fact that $\mathcal{E}$ is trace-preserving (or equivalently $\mathcal{E}^{\dagger}$ is unital).}
As for part (iii), 
since the quantum channel $\mathcal{E}$ is completely positive, it is thus also Hermitian-preserving; so is $\mathcal{E}^\dagger$.
By \lemref{lem::prop_Omega_op} (\ref{lem::prop_Omega_op_hermitian_preserving}), $\Omega_{\sigma}^{\kappa}$ is Hermitian-preserving, thus so is $\left(\Omega_{\sigma}^{\kappa}\right)^{-1}$.  Finally, since the composition of two Hermitian-preserving operators is also Hermitian-preserving,  we conclude that $\Upsilon_{\mathcal{E},\sigma}^{\kappa}$ is Hermitian-preserving. Part (iv) is essentially the data processing inequality; see \eg{} \cite[Thm. II.14]{lesniewski_monotone_1999} and \cite[Thm. 4]{temme_divergence_2010} for the proof. 
\end{proof}

Then 
\begin{equation}
\label{eqn::sdpi_const_rewrite_3}
\begin{split}
 \eta_{\chisq} (\mathcal{E}, \sigma)
 \myeq{\eqref{eqn::sdpi_const}}& \sup_{\rho\in \dset_n:\ \rho\neq \sigma}  \frac{\chisqarg{\mathcal{E}(\rho)}{ \mathcal{E}(\sigma)}}{\chisqarg{\rho}{ \sigma}} \\
 \myeq{\eqref{eqn::sdpi_ratio}} & \sup_{\rho \ge 0:\ \rho\neq \sigma, \tr(\rho) = 1} \frac{\Inner{\rho-\sigma}{\Upsilon_{\mathcal{E},\sigma}^{\kappa}(\rho-\sigma)}_{\Omega_{\sigma}^{\kappa}} 
}{\Inner{\rho-\sigma}{\rho-\sigma}_{\Omega_{\sigma}^{\kappa}}}\\
=&\sup_{A\in \HH_n^0:\ A \neq 0} \frac{\Inner{A}{ \Upsilon_{\mathcal{E},\sigma}^{\kappa}(A)}_{\Omega_{\sigma}^{\kappa}}}{\Inner{A}{A}_{\Omega_{\sigma}^{\kappa}}}. 
\end{split}
\end{equation}

As one might observe, the last equation is closely connected to the 
eigenvalue formalism of the operator
$\Upsilon_{\mathcal{E},\sigma}^{\kappa}$, which is stated in the
following lemma.
\begin{lemma}
\label{lem::sdpi_variational}
For $\sigma\in \dsetplus_n$ and $\kappa\in \mathcal{K}$ and for any quantum channel $\mathcal{E}$ \revise{such that $\mathcal{E}(\sigma)\in \dsetplus_n$}, let $\lambda_2(\Upsilon_{\mathcal{E},\sigma}^{\kappa})$ be the second largest eigenvalue of $\Upsilon_{\mathcal{E},\sigma}^{\kappa}$ (defined in \eqref{eqn::op_Upsilon}).
Then 
\begin{align}
\eta_{\chisq}(\mathcal{E}, \sigma) = \lambda_2(\Upsilon_{\mathcal{E},\sigma}^{\kappa}).
\end{align}
\end{lemma}

\begin{proof}
Since $\Upsilon_{\mathcal{E},\sigma}^{\kappa}$ is positive  semidefinite with respect to the inner product $\Inner{\cdot}{\cdot}_{\Omega_{\sigma}^{\kappa}}$ from \lemref{lem::Upsilon_prop} (\ref{lem::Upsilon_prop_positive}), it admits a spectral decomposition with 
$
\Upsilon_{\mathcal{E},\sigma}^{\kappa}(V_j) = \theta_j V_j,\ \theta_j \ge 0,
$
where $j=1, 2, \cdots, n^2$ and 
$\{V_j\}_{j=1}^{n^2}$ is an orthonormal basis in the Hilbert space $\left(\MM_n, \Inner{\cdot}{\cdot}_{\Omega_{\sigma}^{\kappa}} \right)$. 
Note that $\sigma$ is always an eigenvector of $\Upsilon_{\mathcal{E},\sigma}^{\kappa}$ from \lemref{lem::Upsilon_prop} (\ref{lem::Upsilon_prop_fixed_point});  
without loss of generality, let $V_{1} = \sigma$ and $\theta_{1} = 1$. By the orthogonality of $\{V_j\}_{j}$, we know $0 = \Inner{\sigma}{V_j}_{\Omega_{\sigma}^{\kappa}} = \tr(V_j)$ for $j \ge 2$. 
By \lemref{lem::Upsilon_prop} (\ref{lem::Upsilon_prop_bound_by_one}), $\theta_j \le 1$ for all $1\le j\le n^2$; thus without loss of generality, 
assume $\theta_j$ are listed in descending order and hence $\lambda_2(\Upsilon_{\mathcal{E},\sigma}^{\kappa})  = \theta_2$. 
By rewriting $A = \sum_{j=2}^{n^2} c_j V_j$ in \eqref{eqn::sdpi_const_rewrite_3} where $c_j\in \Complex$, we immediately know that 
$\eta_{\chisq} (\mathcal{E}, \sigma) \le \lambda_2(\Upsilon_{\mathcal{E},\sigma}^{\kappa})$.

By the fact that $\Upsilon_{\mathcal{E},\sigma}^{\kappa}$ is Hermitian-preserving (see \lemref{lem::Upsilon_prop} (\ref{lem::Upsilon_prop_hermitian})), $V_2^\dagger$ is also an eigenvector associated with the eigenvalue 
$\lambda_2(\Upsilon_{\mathcal{E},\sigma}^{\kappa})$. Then we choose $A\in \HH_n^0$ in \eqref{eqn::sdpi_const_rewrite_3} by $\frac{V_2 + V_2^\dagger}{2}$ or $\frac{V_2 - V_2^\dagger}{2i}$. Note that such an $A$ is also an eigenvector of $\Upsilon_{\mathcal{E},\sigma}^{\kappa}$ with the eigenvalue $\lambda_2(\Upsilon_{\mathcal{E},\sigma}^{\kappa})$. 
Then 
$\eta_{\chisq}(\mathcal{E}, \sigma) \ge \Inner{A}{ \Upsilon_{\mathcal{E},\sigma}^{\kappa}(A)}_{\Omega_{\sigma}^{\kappa}} / {\Inner{A}{A}_{\Omega_{\sigma}^{\kappa}}} = \lambda_2(\Upsilon_{\mathcal{E},\sigma}^{\kappa}).$
\end{proof}

\subsection{Another variational formalism of SDPI constants}
Recall the definition of the operator $\mho_{\sigma}^{\kappa}$ from
\eqref{eqn::mho_op_kappa}. In \lemref{lem::sdpi_svd} below, we provide
another variational characterization of the SDPI constant;
essentially, it follows from the connection between the eigenvalue
formalism (as discussed in the last subsection) and the corresponding
singular value formalism. Its classical version is well-known and can
be found in \eg{} the proof of
\cite[Thm. \rom{3}.2]{Raginsky_strong_2016}. This idea for quantum $\chisq$
divergences has appeared implicitly in
\cite[Thm. 9]{temme_divergence_2010}; however, we don't assume
$\sigma$ to be the stationary state of the quantum channel herein,
compared with \cite{temme_divergence_2010}.

\begin{lemma}
\label{lem::sdpi_svd}
Assume that quantum states $\sigma, \mathcal{E}(\sigma) \in \dsetplus_n$.
For any $\kappa\in \mathcal{K}$, 
\begin{align}
\label{eqn::sdpi_max_corr}
\sqrt{\eta_{\chisq}(\mathcal{E}, \sigma)} = \max_{F, G} \Abs{ \Inner{\Kop(F) }{G}_{\mho_{\mathcal{E}(\sigma)}^{\kappa}}}, 
\end{align}
where the operator $\Kop$ is defined by
\begin{align}
\label{eqn::op_K}
\Kop := \Gamma_{\mathcal{E}(\sigma)}^{-1}  \circ \mathcal{E} \circ \Gamma_{\sigma},
\end{align}
and the maximum is taken over all $F, G\in \MM_n$ such that
\begin{align}
\label{eqn::sdpi_max_corr_cond}
\Inner{\unit_n}{F}_{\mho_{\sigma}^{\kappa}} = \Inner{\unit_n}{G}_{\mho_{\mathcal{E}(\sigma)}^{\kappa}} = 0, \qquad \Norm{F}_{\mho_{\sigma}^{\kappa}} = \Norm{G}_{\mho_{\mathcal{E}(\sigma)}^{\kappa}} = 1.
\end{align}
\end{lemma}

\begin{proof}[Proof of \lemref{lem::sdpi_svd}]  
First, we rewrite \lemref{lem::sdpi_variational} in the language of the relative density (whose classical analog is the Radon–Nikodym derivative); 
specifically, to get the third equality below, $A$ is replaced by $\Gamma_{\sigma}(A)$. 
By \lemref{lem::sdpi_variational},
\begin{align*}
&\eta_{\chisq}(\mathcal{E}, \sigma) = \sup_{A\in \MM_n^0,\ A \neq 0} \frac{\Inner{A}{\Upsilon_{\mathcal{E},\sigma}^{\kappa} (A)}_{\Omega_{\sigma}^{\kappa}}}{\Inner{A}{A}_{\Omega_{\sigma}^{\kappa}}} 
=\sup_{A\in \MM_n^0,\ A \neq 0} \frac{\Inner{\mathcal{E} (A)} {\Omega_{\mathcal{E}(\sigma)}^{\kappa}\circ \mathcal{E} (A)}_{\hs}}{\Inner{A}{\Omega_{\sigma}^{\kappa} (A)}_{\hs}}\\
=& \sup_{A\neq 0,\ \Inner{\unit_n}{A}_{\mho_{\sigma}^{\kappa}} = 0} \frac{\Inner{\mathcal{E} \circ \Gamma_{\sigma} (A)} {\Omega_{\mathcal{E}(\sigma)}^{\kappa}\circ \mathcal{E} \circ \Gamma_{\sigma} (A)}_{\hs}}{\Inner{\Gamma_{\sigma} (A)}{\Omega_{\sigma}^{\kappa} \circ \Gamma_{\sigma} (A)}_{\hs}} 
\myeq{\eqref{eqn::mho_op_kappa}} \sup_{A\neq 0,\ \Inner{\unit_n}{A}_{\mho_{\sigma}^{\kappa}} = 0} \frac{\Inner{\Kop (A)} {\mho_{\mathcal{E}(\sigma)}^{\kappa} \circ \Kop (A)}_{\hs}}{\Inner{A}{A}_{\mho_{\sigma}^{\kappa}}}\\
=& \sup_{A\neq 0,\ \Inner{\unit_n}{A}_{\mho_{\sigma}^{\kappa}} = 0} \frac{\Inner{A} {\left(\mho_{\sigma}^{\kappa} \right)^{-1}\circ\Kop^{\dagger} \circ \mho_{\mathcal{E}(\sigma)}^{\kappa} \circ \Kop (A)}_{\mho_{\sigma}^{\kappa}}}{\Inner{A}{A}_{\mho_{\sigma}^{\kappa}}}.
\end{align*}

As for the operator $\Kop$, it can be straightforwardly checked that 
\begin{itemize}
\item $\Kop$ is completely positive and unital ($\Kop(\unit_n) = \unit_n$). 

\item $\Kop^{\dagger} = \Gamma_{\sigma} \circ \mathcal{E}^{\dagger} \circ \Gamma_{\mathcal{E}(\sigma)}^{-1}$ is completely positive, trace-preserving, and $\Kop^{\dagger}(\mathcal{E}(\sigma)) = \sigma$.

\item Consider the following two Hilbert spaces $\mathscr{H}_1$ and $\mathscr{H}_2$, 
\begin{align*}
\mathscr{H}_1 &:= \left\{A\in \MM_n: \Inner{\unit_n}{A}_{\mho_{\sigma}^{\kappa}} = 0 \right\}, \text{ equipped with the inner product } \Inner{\cdot}{\cdot}_{\mho_{\sigma}^{\kappa}};\\
\mathscr{H}_2 &:= \left\{A\in \MM_n: \Inner{\unit_n}{A}_{\mho_{\mathcal{E}(\sigma)}^{\kappa}} = 0 \right\}, \text{ equipped with the inner product } \Inner{\cdot}{\cdot}_{\mho_{\mathcal{E}(\sigma)}^{\kappa}}. 
\end{align*}
Then we can readily verify that $\Kop$ is an operator from $\mathscr{H}_1$ to $\mathscr{H}_2$, \ie, if $\Inner{\unit_n}{A}_{\mho_{\sigma}^{\kappa}} = 0$, then $\Inner{\unit_n}{\Kop(A)}_{\mho_{\mathcal{E}(\sigma)}^{\kappa}} = 0$. The dual operator of $\Kop$, denoted by $\wt{\Kop}$, maps from $\mathscr{H}_2$ to $\mathscr{H}_1$ and it is explicitly given by 
$\wt{\Kop} = \left(\mho_{\sigma}^{\kappa}\right)^{-1} \circ \Kop^{\dagger}\circ \mho_{\mathcal{E}(\sigma)}^{\kappa}$. 
\end{itemize}
Then, we have
\begin{align*}
\eta_{\chisq}(\mathcal{E}, \sigma) = \sup_{A\neq 0,\ A\in \mathscr{H}_1} \frac{\Inner{A} {\wt{\Kop} \circ \Kop (A)}_{\mho_{\sigma}^{\kappa}}}{\Inner{A}{A}_{\mho_{\sigma}^{\kappa}}}. 
\end{align*}

Let us denote the SVD decomposition of $\wt{\Kop}$ by $\wt{\Kop}(\cdot) = \sum_{j} a_j \phi_j \Inner{\varphi_j}{\cdot}_{\mho_{\mathcal{E}(\sigma)}^{\kappa}}$ where $a_j \ge 0$, $\{\phi_j\}_{j}$ and $\{\varphi_j\}_{j}$ are orthonormal basis of $\mathscr{H}_1$ and $\mathscr{H}_2$ respectively. Then, easily we know $\Kop(\cdot) = \sum_{j} a_j \varphi_j \Inner{\phi_j}{\cdot}_{\mho_{\sigma}^{\kappa}}$ and that $\wt{\Kop}\circ \Kop (\cdot) = \sum_{j} a_j^2 \phi_j \Inner{\phi_j}{\cdot}_{\mho_{\sigma}^{\kappa}}$.  Then $\eta_{\chisq}(\mathcal{E}, \sigma)$ is simply the largest value of $a_j^2$; namely, $\sqrt{\eta_{\chisq}(\mathcal{E}, \sigma)}$ is the largest singular value of $\Kop$, and the result in \lemref{lem::sdpi_svd} follows immediately. 
\end{proof}

\subsection{Comparison of SDPI constants}
\label{subsec::comparison}

First, we provide a uniform lower bound of $\sqrt{\eta_{\chisq}(\mathcal{E}, \sigma)}$ for any $\kappa\in \mathcal{K}$ in terms of $\eta_{\chi^2_{\kappa_{1/2}}}(\mathcal{E}, \sigma)$ in \lemref{lem::chialphahalf_less_chihalfsquare}, which is a new result to the best of our knowledge. 
One of our corollaries in \eqref{eqn::cc_half_alpha} can also be derived by \cite[Thm. 4.4]{hiai_contraction_2016} and \cite[Thm. 5.3]{hiai_contraction_2016}.
However, our approach to show \eqref{eqn::cc_half_alpha} is different from \cite{hiai_contraction_2016}:  their result comes from comparing the contraction coefficient $\eta_{\chisq}(\mathcal{E})$ with $\eta_{\tr}(\mathcal{E})$ (the contraction coefficient for trace norm); we use the SDPI constant of the Petz recovery map as the bridge. 
Second, we consider quantum-classical (QC) channels and provide the ordering of SDPI constants for different $\kappa$ in \lemref{lem::monotonicity}; similar results have appeared in \cite[Prop. 5.5]{hiai_contraction_2016} for contraction coefficients. 

\begin{lemma}
\label{lem::chialphahalf_less_chihalfsquare}
For any quantum channel $\mathcal{E}$ and quantum state $\sigma\in \dsetplus_n$ such that $\mathcal{E}(\sigma)\in \dsetplus_n$, we have
\begin{align}
\label{eqn::chialphahalf_less_chihalfsquare}
\eta_{\chi^2_{\kappa_{1/2}}}(\mathcal{E}, \sigma) \le \sqrt{\eta_{\chisq}(\mathcal{E}, \sigma) \eta_{\chisq}(\mathcal{R}_{\mathcal{E}, \sigma}, \mathcal{E}(\sigma))} \le \sqrt{\eta_{\chisq}(\mathcal{E}, \sigma)},
\end{align}
where $\mathcal{R}_{\mathcal{E}, \sigma}$ is the Petz recovery map, defined by  
\begin{align}
\label{eqn::petz_recovery}
\mathcal{R}_{\mathcal{E}, \sigma}(A)  := 
\sigma^{\half} \mathcal{E}^{\dagger} (\mathcal{E}(\sigma)^{-\half} A \mathcal{E}(\sigma)^{-\half}) \sigma^{\half} \equiv \Gamma_{\sigma} \circ \mathcal{E}^{\dagger} \circ \Gamma^{-1}_{\mathcal{E}(\sigma)}(A), \qquad \forall A\in \MM_n
\end{align}
mapping $\mathcal{E}(\sigma)$ to $\sigma$. 
\end{lemma}

The followings are immediate consequences
\revise{of the lemma above}.
\begin{coro}
Under the same assumption as in \lemref{lem::chialphahalf_less_chihalfsquare}, 
\begin{enumerate}[(i)]
\item The SDPI constant associated with $\kappa_{1/2}$ for the pair $(\mathcal{E}, \sigma)$ equals the SDPI constant for the recovery map pair $(\mathcal{R}_{\mathcal{E}, \sigma}, \mathcal{E}(\sigma))$, that is to say,
\begin{align*}
\eta_{\chi^2_{\kappa_{1/2}}}(\mathcal{E}, \sigma) = \eta_{\chi^2_{\kappa_{1/2}}}(\mathcal{R}_{\mathcal{E}, \sigma}, \mathcal{E}(\sigma)).
\end{align*}

\item Further assume that for any $\sigma\in \dsetplus_n$, we have $\mathcal{E}(\sigma)\in \dsetplus_n$. Then, for the contraction coefficient of the quantum channel $\mathcal{E}$, we have
\begin{align}
\label{eqn::cc_half_alpha}
\eta_{\chi^2_{\kappa_{1/2}}}(\mathcal{E}) \le \sqrt{\eta_{\chisq}(\mathcal{E})}.
\end{align}
\end{enumerate}
\end{coro}

\begin{proof}
The first part comes from letting $\kappa = \kappa_{1/2}$ in \eqref{eqn::chialphahalf_less_chihalfsquare} and the fact that the Petz recovery map of $\mathcal{R}_{\mathcal{E}, \sigma}$ is exactly the channel $\mathcal{E}$; the second part comes from taking the supremum over all $\sigma\in \dsetplus_n$.
\end{proof}

\begin{proof}[Proof of \lemref{lem::chialphahalf_less_chihalfsquare}]
It is straightforward to verify that $\mathcal{R}_{\mathcal{E}, \sigma}$, defined in \eqref{eqn::petz_recovery}, is a bona-fide quantum channel, mapping the quantum state $\mathcal{E}(\sigma)$ back to $\sigma$.
We can easily verify by definition \eqref{eqn::op_Upsilon} and \eqref{eqn::petz_recovery} that 
\begin{align}
\label{eqn::Upsilon_op_recovery_map}
\Upsilon_{\mathcal{E},\sigma}^{\kappa_{1/2}} = \mathcal{R}_{\mathcal{E}, \sigma} \circ \mathcal{E}.
\end{align}

Recall from \lemref{lem::sdpi_variational} that there exists a $\lambda_2 \equiv \lambda_2(\Upsilon_{\mathcal{E},\sigma}^{\kappa_{1/2}}) = \eta_{\chi^2_{\kappa_{1/2}}}(\mathcal{E}, \sigma)$ and a traceless Hermitian matrix $V\in \HH^0_{n}$ such that $\Upsilon_{\mathcal{E},\sigma}^{\kappa_{1/2}}(V) = \lambda_2 V$. 
Let $\wt{V} := \mathcal{E}(V)\in \HH_n^0$. Then 
\begin{align*}
 \left(\eta_{\chi^2_{\kappa_{1/2}}}(\mathcal{E}, \sigma)\right)^2 &= \lambda_2^2  = \frac{\Inner{\Upsilon_{\mathcal{E},\sigma}^{\kappa_{1/2}}(V)}{\Upsilon_{\mathcal{E},\sigma}^{\kappa_{1/2}}(V)}_{\Omega_{\sigma}^{\kappa} } }{\Inner{V}{V}_{\Omega_{\sigma}^{\kappa} }} \\
&\myeq{\eqref{eqn::Upsilon_op_recovery_map}} \frac{\Inner{\mathcal{R}_{\mathcal{E}, \sigma} (\wt{V})}{\Omega_{\sigma}^{\kappa} \circ \mathcal{R}_{\mathcal{E}, \sigma} (\wt{V})}_{\hs} }{\Inner{V}{V}_{\Omega_{\sigma}^{\kappa} }} \\
&= \frac{\Inner{\mathcal{R}_{\mathcal{E}, \sigma} (\wt{V})}{\Omega_{\mathcal{R}_{\mathcal{E}, \sigma}(\mathcal{E}(\sigma))}^{\kappa} \circ \mathcal{R}_{\mathcal{E}, \sigma} (\wt{V})}_{\hs} }{\Inner{\wt{V}}{\Omega_{\mathcal{E}(\sigma)}^{\kappa} (\wt{V})}_{\hs} } \frac{\Inner{\wt{V}}{\Omega_{\mathcal{E}(\sigma)}^{\kappa} (\wt{V})}_{\hs} }{\Inner{V}{V}_{\Omega_{\sigma}^{\kappa} }} \\
& \le \eta_{\chisq}\left(\mathcal{R}_{\mathcal{E}, \sigma}, \mathcal{E}(\sigma)\right) \eta_{\chisq}(\mathcal{E}, \sigma). 
\end{align*}
The inequality in the last step follows from \lemref{lem::sdpi_variational}. Hence, we have proved the first inequality in \eqref{eqn::chialphahalf_less_chihalfsquare}; the second inequality follows immediately from the data processing inequality of the quantum $\chisq$ divergence. 
\end{proof}

Next, we consider any quantum-classical (QC) channel $\mathcal{E}$, which refers to a physical process in which one first performs a measurement according to a POVM $\{F_{j}\}_{j=1}^{n}$ ($F_j\in \MM_n$ are positive semidefinite and $\sum_{j=1}^{n} F_j = \unit_{n}$); then based on the measurement outcome, one prepares a pure state, selected from a set $\{\psi_j\}_{j=1}^{n}$ which also forms an orthonormal basis of $\hbt$. More specifically, 
\begin{align}
\label{eqn::qc_channel}
\mathcal{E}(A) = \sum_{j=1}^{n} \tr\left(F_j A\right) \Ket{\psi_j} \Bra{\psi_j}, \qquad \forall A\in \MM_n.  
\end{align}
Define a ratio $\sdpiratio_{\mathcal{E},\sigma}^{\kappa}$ on $\HH^0_n$ by 
\begin{align}
\label{eqn::sdpi_ratio_R}
\sdpiratio_{\mathcal{E},\sigma}^{\kappa} (A) :=  
\frac{\Inner{\mathcal{E}(A)}{\Omega_{\mathcal{E}(\sigma)}^{\kappa} \circ  \mathcal{E}(A)}_{\hs} 
}{\Inner{A}{\Omega_{\sigma}^{\kappa} (A)}_{\hs}} \equiv \frac{\Inner{A}{\Upsilon_{\mathcal{E},\sigma}^{\kappa} (A)}_{\Omega_{\sigma}^{\kappa}}}{\Inner{A}{A}_{\Omega_{\sigma}^{\kappa}}}, \qquad \text{for } A\in \HH^0_{n},\ A \neq 0.
\end{align}

\begin{lemma}
\label{lem::monotonicity}
Suppose $\kappa \ge \kappa_{1/2}$, $\mathcal{E}$ is a  QC channel \revise{with $F_j\neq 0$ for all $1\le j\le n$} and $\sigma\in \dsetplus_n$. Then 
\begin{align}
\label{eqn::monotonicity_sdpiratio}
\sdpiratio_{\mathcal{E}, \sigma}^{\kappa}(A) \le \sdpiratio_{\mathcal{E}, \sigma}^{\kappa_{1/2}}(A),\qquad \forall A \in \HH_n^0,\ A\neq 0. 
\end{align}
Consequently, we have 
\begin{align}
\label{eqn::monotonicity_sdpi_constant}
\eta_{\chisq}(\mathcal{E}, \sigma) \le \eta_{\chi^2_{\kappa_{1/2}}}(\mathcal{E}, \sigma).
\end{align}
\end{lemma}

\begin{proof}
By \eqref{eqn::qc_channel} and \eqref{eqn::AA_omega_kappa}, we can readily calculate that 
\begin{align}
\label{eqn::qc_value}
\Inner{\mathcal{E}(A)}{\Omega_{\mathcal{E}(\sigma)}^{\kappa} (\mathcal{E}(A))}_{\hs} = \sum_{j=1}^{n} \frac{\Abs{\tr(F_j A)}^2}{\tr(F_j \sigma)}, 
\end{align}
which is independent of $\kappa$. 
By \eqref{eqn::AA_omega_kappa}, it is straightforward to observe that when $\kappa \ge \kappa_{1/2}$, one has $\Inner{A}{\Omega_{\sigma}^{\kappa}(A)}_{\hs} \ge \Inner{A}{\Omega_{\sigma}^{\kappa_{1/2}}(A)}_{\hs}$. Thus \eqref{eqn::monotonicity_sdpiratio} follows immediately; \eqref{eqn::monotonicity_sdpi_constant} follows from \eqref{eqn::monotonicity_sdpiratio} by taking the supremum over all non-zero $A\in \HH^0_n$ (see \eqref{eqn::sdpi_const_rewrite_3}). 
\end{proof}

\section{Proof of \thmref{thm::chisq_tensorize}}
\label{sec::chisq_tensorize}

{\bf \noindent Setting up:} 
First notice that it is sufficient to prove \thmref{thm::chisq_tensorize} for $N = 2$. The general case can be straightforwardly  proved by mathematical induction on $N$. Next, for the case $N=2$, one direction is trivial: suppose $\rho_1$ achieves the maximum in $\eta_{\chisq}(\mathcal{E}_1, \sigma_1)$; let $\rho_{1,2} = \rho_1 \otimes \sigma_2$ and by direct calculation,
\begin{align*}
\eta_{\chisq}(\mathcal{E}_1\otimes \mathcal{E}_2, \sigma_1 \otimes \sigma_2) &\ge \frac{\chisqarg{\mathcal{E}_1\otimes \mathcal{E}_2(\rho_{1,2})}{\mathcal{E}_1\otimes \mathcal{E}_2 (\sigma_1\otimes \sigma_2)}}{\chisqarg{\rho_{1,2}}{\sigma_1\otimes \sigma_2}} \\
&\myeq{\eqref{eqn::Omega_op_composite}} \frac{\chisqarg{\mathcal{E}_1(\rho_1)}{\mathcal{E}_1(\sigma_1)}}{\chisqarg{\rho_1}{\sigma_1}} = \eta_{\chisq}\left(\mathcal{E}_1, \sigma_1 \right).
\end{align*}
Similarly, by choosing $\rho_{1,2} = \sigma_1\otimes \rho_2$ where $\rho_2$ achieves the maximum in $\eta_{\chisq}(\mathcal{E}_2, \sigma_2)$, we have $\eta_{\chisq}(\mathcal{E}_1\otimes \mathcal{E}_2, \sigma_1 \otimes \sigma_2) \ge \eta_{\chisq}\left(\mathcal{E}_2, \sigma_2 \right)$. Therefore,
\begin{align*}
\eta_{\chisq}(\mathcal{E}_1\otimes \mathcal{E}_2, \sigma_1 \otimes \sigma_2) \ge \max\left(\eta_{\chisq}\left(\mathcal{E}_1, \sigma_1 \right), \eta_{\chisq}\left(\mathcal{E}_2, \sigma_2 \right)\right). 
\end{align*}

In the below, we shall prove the other direction, \ie, 
\begin{align}
\label{eqn::tensorize_less}
\eta_{\chisq}(\mathcal{E}_1\otimes \mathcal{E}_2, \sigma_1 \otimes \sigma_2) \le \max\left(\eta_{\chisq}\left(\mathcal{E}_1, \sigma_1 \right), \eta_{\chisq}\left(\mathcal{E}_2, \sigma_2 \right)\right) =: \eta_{\max}^{\kappa}. 
\end{align}

\medskip

{\bf \noindent Notations:} 

Since we fix states $\sigma_m$ and channels $\mathcal{E}_m$ for $m = 1, 2$ throughout this section, 
let us denote $\Upsilon^{\kappa}_{m} \equiv \Upsilon_{\mathcal{E}_m, \sigma_m}^{\kappa}$ for simplicity of notation.
By \lemref{lem::sdpi_variational}, 
$\Upsilon_{m}^{\kappa}$ has an eigen-basis $\{V^{\kappa, m}_j \}_{j=1}^{n_m^2}$ associated with eigenvalue $\{\theta^{\kappa, m}_{j}\}_{j=1}^{n_m^2}$ with respect to the inner product $\Inner{\cdot}{\cdot}_{\Omega_{\sigma_m}^{\kappa}}$ such that 
\begin{align*}
\Upsilon^{\kappa}_{m}(V_j^{\kappa, m}) = \theta_{j}^{\kappa, m} V_j^{\kappa, m}, \qquad 1\le j \le n_m^2,
\end{align*}
where $V^{\kappa, m}_{1} = \sigma_m$, $\theta_{1}^{\kappa, m} = 1$ and
$V_j^{\kappa, m}$ are Hermitian for all $1\le j \le n_m^2$, since from
\lemref{lem::Upsilon_prop} $\Upsilon^{\kappa}_{m}$ are both
Hermitian-preserving positive semidefinite operators.  In addition, we
know from \lemref{lem::sdpi_variational} (or say
\lemref{lem::Upsilon_prop} (\ref{lem::Upsilon_prop_bound_by_one}))
that for both $m = 1, 2$,
\begin{align*}
\theta_j^{\kappa, m} \le \eta_{\chisq}(\mathcal{E}_m, \sigma_m) \le \eta_{\max}^{\kappa}, \qquad \text{for any } 2\le j \le n_m^2. 
\end{align*}

For convenience, let $\sigma = \sigma_1 \otimes \sigma_2$ and $\mathcal{E} = \mathcal{E}_1 \otimes \mathcal{E}_2$; let $\Upsilon^{\kappa} \equiv \Upsilon_{\mathcal{E},\sigma}^{\kappa}$. 
For any index pair $\vect{J} = (j_1, j_2)$, define
\begin{align*}
V_{\vect{J}}^{\kappa} := V_{j_1}^{\kappa, 1}\otimes V_{j_2}^{\kappa, 2},\qquad 
\theta_{\vect{J}}^{\kappa}  := \theta_{j_1}^{\kappa, 1} \theta_{j_2}^{\kappa, 2}.
\end{align*}

\medskip

{\bf \noindent Case (\rom{1}): For $\kappa = \kappa_{1/2}$ and any quantum channel.} From \lemref{lem::prop_Omega_op_composite} part (\ref{lem::prop_Omega_op_composite_tensorize}), $\Omega_{\sigma}^{\kappa}$ tensorizes, thus $\Upsilon^{\kappa} = \Upsilon^{\kappa}_{1}\otimes \Upsilon_{2}^{\kappa}$. Next, we can straightforwardly verify that $\left\{ V_{\vect{J}}^{\kappa}\right\}_{\vect{J}}$ (for $\vect{J} = (j_1, j_2)$) is an orthonormal eigenbasis of $\Upsilon^{\kappa}$ with respect to the inner product $\Inner{\cdot}{\cdot}_{\Omega_{\sigma}^{\kappa}}$, and the associated eigenvalues are $\left\{\theta_{\vect{J}}^{\kappa}\right\}_{\vect{J}}$. 
The largest eigenvalue of $\Upsilon^{\kappa}$ on the domain $\text{span}(\sigma)^{\perp} \equiv \MM_{n_1\times n_2}^0$ becomes $\max_{\vect{J}\neq (1, 1)}\{\theta_{\vect{J}}^{\kappa}\} = \eta_{\max}^{\kappa}$. Therefore, by \lemref{lem::sdpi_variational}, we have
$
\eta_{\chisq}(\mathcal{E}, \sigma) = \max_{\vect{J}\neq (1,1)}\{\theta_{\vect{J}}^{\kappa}\} = \max\left(\eta_{\chisq}(\mathcal{E}_1, \sigma_1), \eta_{\chisq}(\mathcal{E}_2, \sigma_2) \right).
$
Thus we complete the proof of \eqref{eqn::tensorize_less} for the case $\kappa_{1/2}$. 

\medskip

{\bf \noindent Case (\rom{2}): For $\kappa \ge \kappa_{1/2}$ and QC channels.}
Let us decompose $\vect{A}\in \HH^0_{n_1\times n_2}$ by $\vect{A} = \sum_{\vect{J}} c_{\vect{J}} V_{\vect{J}}^{\kappa}$ where $c_{\vect{J}}\in \Real$. From the constraint that $\tr(\vect{A}) = 0$, we know $c_{(1,1)} = 0$. Thus, we can rewrite $\vect{A}$ by 
\begin{align}
\label{eqn::H_decomp}
\vect{A} = \sigma_1 \otimes A_2 + A_1 \otimes \sigma_2 + \wt{A},
\end{align}
where 
\begin{equation}
\left\{
\begin{split}
A_2 &= \sum_{\vect{J}:\ j_1=1,\ 2\le j_2 \le n_2^2} c_{\vect{J}} V^{\kappa, 2}_{j_2};\\
A_1 &= \sum_{\vect{J}:\ 2 \le j_1 \le n_1^2,\ j_2 = 1} c_{\vect{J}} V^{\kappa, 1}_{j_1};\\
\wt{A} &= \sum_{\vect{J}:\ j_1 \neq 1,\ j_2 \neq 1} c_{\vect{J}} V_{\vect{J}}.
\end{split}\right.
\end{equation}
To prove \eqref{eqn::tensorize_less}, by \eqref{eqn::sdpi_const_rewrite_3}, it is equivalent to prove that for all $\vect{A}\in \HH_{n_1\times n_2}^{0}$ and $\vect{A} \neq 0$, we have 
\begin{align}
\label{eqn::generic_alpha_suff_cond}
\frac{\Inner{\vect{A}}{ \Upsilon^{\kappa}(\vect{A})}_{\Omega_{\sigma}^{\kappa}}}{\Inner{\vect{A}}{\vect{A}}_{\Omega_{\sigma}^{\kappa}}}  \le \eta_{\max}^{\kappa}.
\end{align}

The next lemma shows that it is sufficient to consider $\vect{A}$ as $\wt{A}$. 
\begin{lemma}
\label{lem::tilde_H}
If \eqref{eqn::generic_alpha_suff_cond} holds for any $\vect{A} \in \HH_{n_1}^0 \otimes \HH_{n_2}^0$, then \eqref{eqn::generic_alpha_suff_cond} holds for any $\vect{A}\in \HH_{n_1\times n_2}^{0}$.
\end{lemma}
Notice that $\HH_{n_1}^0 \otimes \HH_{n_2}^0 \subset \HH_{n_1\times n_2}^0$. 
The proof of this lemma is postponed to the end of this section and let us continue to complete the proof of \thmref{thm::chisq_tensorize}.
 It is straightforward to verify that when $\mathcal{E}_1$ and $\mathcal{E}_2$ are QC channels, $\mathcal{E} = \mathcal{E}_1 \otimes \mathcal{E}_2$ is also a QC channel for the composite system. 
By \lemref{lem::monotonicity}, for any $\vect{A}\in \HH^0_{n_1} \otimes \HH^0_{n_2}$, we have 
\begin{align*}
\frac{\Inner{\vect{A}}{\Upsilon^{\kappa} (\vect{A})}_{\Omega_{\sigma}^{\kappa}}}{\Inner{\vect{A}}{\vect{A}}_{\Omega_{\sigma}^{\kappa}}}
&\myle{\eqref{eqn::monotonicity_sdpiratio}} \frac{\Inner{\vect{A}}{\Upsilon^{\kappa_{1/2}} (\vect{A})}_{\Omega_{\sigma}^{{\kappa_{1/2}}}}}{\Inner{\vect{A}}{\vect{A}}_{\Omega_{\sigma}^{{\kappa_{1/2}}}}}\\
&\le \eta_{\chi^2_{{\kappa_{1/2}}}}(\mathcal{E}_1, \sigma_1) \eta_{\chi^2_{{\kappa_{1/2}}}}(\mathcal{E}_2, \sigma_2)\\
&\le \max\left( \bigl(\eta_{\chi^2_{{\kappa_{1/2}}}}(\mathcal{E}_1, \sigma_1)\bigr)^2, \bigl(\eta_{\chi^2_{{\kappa_{1/2}}}}(\mathcal{E}_2, \sigma_2)\bigr)^2\right) \\
& \myle{\eqref{eqn::chialphahalf_less_chihalfsquare}} \max\left(\eta_{\chisq}(\mathcal{E}_1, \sigma_1), \eta_{\chisq}(\mathcal{E}_2, \sigma_2) \right) = \eta_{\max}^{\kappa}.
\end{align*}
The second inequality comes from the observation that $\Upsilon^{\kappa_{1/2}}$ is a positive semidefinite operator on the space $\HH_{n_1}^0 \otimes \HH_{n_2}^0$ with eigenvalues $\theta^{\kappa_{1/2}}_{\vect{J}}$; for $j_1\neq 1$ and $j_2\neq 1$, recall from previous results that $\theta^{{\kappa_{1/2}}}_{\vect{J}} = \theta^{{\kappa_{1/2}}, 1}_{j_1} \theta^{{\kappa_{1/2}},2}_{j_2} \le \eta_{\chi^2_{\kappa_{1/2}}}(\mathcal{E}_1, \sigma_1) \eta_{\chi^2_{{\kappa_{1/2}}}}(\mathcal{E}_2, \sigma_2)$. 
The last equation means \eqref{eqn::generic_alpha_suff_cond} holds for all Hermitian $\vect{A}\in \HH_{n_1}^0 \otimes \HH_{n_2}^0$ and by \lemref{lem::tilde_H}, \eqref{eqn::generic_alpha_suff_cond} holds for all Hermitian $\vect{A}\in \HH_{n_1\times n_2}^0$. This completes the proof of \thmref{thm::chisq_tensorize}.

\begin{proof}[Proof of \lemref{lem::tilde_H}.]

For any Hermitian $\vect{A}$ in \eqref{eqn::H_decomp}, 
we claim that 
\begin{align}
\label{eqn::H_Upsilon_H_expansion}
\begin{split}
\Inner{\vect{A}}{ \Upsilon^{\kappa}(\vect{A})}_{\Omega_{\sigma}^{\kappa}} = \Inner{\sigma_1 \otimes A_2}{ \Upsilon^{\kappa}(\sigma_1 \otimes A_2)}_{\Omega_{\sigma}^{\kappa}} &+ \Inner{A_1\otimes \sigma_2}{ \Upsilon^{\kappa}(A_1\otimes \sigma_2)}_{\Omega_{\sigma}^{\kappa}} \\
&+ \Inner{\wt{A}}{ \Upsilon^{\kappa}(\wt{A})}_{\Omega_{\sigma}^{\kappa}}.
\end{split}
\end{align}

To prove this, we need to show that all cross product terms in the expansion of $\Inner{\vect{A}}{\Upsilon^{\kappa}(\vect{A})}_{\Omega_{\sigma}^{\kappa}}$ vanish. For instance, consider any $\vect{B}\in \MM_{n_1}\otimes \MM_{n_2}$, 
\begin{align*}
& \Inner{\vect{B}}{\Upsilon^{\kappa}(\sigma_1\otimes A_2)}_{\Omega_{\sigma}^{\kappa}} \\
\myeq{\eqref{eqn::op_Upsilon}}& \Inner{\vect{B}}{\mathcal{E}^{\dagger} \circ \Omega_{\mathcal{E}(\sigma)}^{\kappa}\circ \mathcal{E} (\sigma_1 \otimes A_2)}_{\hs}\\
=& \Inner{\vect{B}}{\mathcal{E}^{\dagger} \circ \Omega_{\mathcal{E}(\sigma)}^{\kappa}\left(\mathcal{E}_1 (\sigma_1) \otimes \mathcal{E}_2(A_2)\right)}_{\hs} \\
\myeq{\eqref{eqn::Omega_op_composite}}& \Inner{\vect{B}}{\mathcal{E}^{\dagger} \left(\unit_{n_1} \otimes \Omega_{\mathcal{E}_2(\sigma_2)}^{\kappa}\circ \mathcal{E}_2(A_2)\right)}_{\hs} \\
=& \Inner{\vect{B}}{\unit_{n_1} \otimes \left(\mathcal{E}_2^{\dagger}\circ \Omega_{\mathcal{E}_2(\sigma_2)}^{\kappa}\circ \mathcal{E}_2(A_2)\right)}_{\hs}.
\end{align*}
If $\vect{B} = A_1\otimes \sigma_2$ or $\vect{B} = \wt{A}$, by plugging the expression of $A_1$ or $\wt{A}$ into the last equation and after expanding all terms, it is straightforward to verify that $\Inner{\vect{B}}{\Upsilon^{\kappa}(\sigma_1\otimes A_2)}_{\Omega_{\sigma}^{\kappa}} = 0$ for both choices of $\vect{B}$. We can apply similar arguments to $\Inner{\vect{B}}{\Upsilon^{\kappa}(A_1 \otimes \sigma_2)}_{\Omega_{\sigma}^{\kappa}}$ for $\vect{B} = \sigma_1\otimes A_2$ or $\vect{B} = \wt{A}$. 
Similarly, we have (or let $\mathcal{E} = \id_{n_1}\otimes \id_{n_2}$ in \eqref{eqn::H_Upsilon_H_expansion})
\begin{equation}
\label{eqn::hh}
\begin{split}
\Inner{\vect{A}}{\vect{A}}_{\Omega_{\sigma}^{\kappa}} 
&= \Inner{\sigma_1\otimes A_2}{\sigma_1\otimes A_2}_{\Omega_{\sigma}^{\kappa}} + \Inner{A_1\otimes \sigma_2}{A_1\otimes \sigma_2}_{\Omega_{\sigma}^{\kappa}} + \Inner{\wt{A}}{\wt{A}}_{\Omega_{\sigma}^{\kappa}}\\
&\myeq{\eqref{eqn::Omega_op_composite}} \Inner{A_2}{A_2}_{\Omega_{\sigma_2}^{\kappa}} + \Inner{A_1}{A_1}_{\Omega_{\sigma_1}^{\kappa}} + \Inner{\wt{A}}{\wt{A}}_{\Omega_{\sigma}^{\kappa}}.
\end{split}
\end{equation}

Let us simplify the term on the right hand side of \eqref{eqn::H_Upsilon_H_expansion}. For instance, 
\begin{align*}
& \Inner{\sigma_1 \otimes A_2}{ \Upsilon^{\kappa}(\sigma_1 \otimes A_2)}_{\Omega_{\sigma}^{\kappa}} \\
=& \Inner{\mathcal{E}(\sigma_1 \otimes A_2)}{\Omega_{\mathcal{E}(\sigma)}^{\kappa} \circ \mathcal{E}(\sigma_1 \otimes A_2)}_{\hs}  \\
\myeq{\eqref{eqn::Omega_op_composite}}&  \Inner{\mathcal{E}_1(\sigma_1) \otimes \mathcal{E}_2(A_2)}{\unit_{n_1}\otimes \Omega_{\mathcal{E}_2(\sigma_2)}^{\kappa} \circ \mathcal{E}_2 (A_2)}_{\hs} \\
=& \Inner{\mathcal{E}_2(A_2)}{\Omega_{\mathcal{E}_2(\sigma_2)}^{\kappa} \circ \mathcal{E}_2 (A_2)}_{\hs} \\
=& \Inner{A_2}{\Upsilon^{\kappa}_{2} (A_2)}_{\Omega_{\sigma_2}^{\kappa}}\\
\le&\ \eta_{\chisq}(\mathcal{E}_2, \sigma_2) \Inner{A_2}{A_2}_{\Omega_{\sigma_2}^{\kappa}} \\
\le&\ \eta_{\max}^{\kappa} \Inner{A_2}{A_2}_{\Omega_{\sigma_2}^{\kappa}}.
\end{align*}
Similarly, 
\begin{align*}
\Inner{A_1\otimes \sigma_2}{ \Upsilon^{\kappa}(A_1\otimes \sigma_2)}_{\Omega_{\sigma}^{\kappa}} \le \eta_{\max}^{\kappa} \Inner{A_1}{A_1}_{\Omega_{\sigma_1}^{\kappa}}. 
\end{align*}
Therefore, we have
\begin{align}
\label{eqn::hh_upsilon}
\begin{aligned}
\Inner{\vect{A}}{ \Upsilon^{\kappa}(\vect{A})}_{\Omega_{\sigma}^{\kappa}} \le &  \eta_{\max}^{\kappa} \left(\Inner{A_2}{A_2}_{\Omega_{\sigma_2}^{\kappa}} + \Inner{A_1}{A_1}_{\Omega_{\sigma_1}^{\kappa}} \right) 
+ \Inner{\wt{A}}{ \Upsilon^{\kappa}(\wt{A})}_{\Omega_{\sigma}^{\kappa}}.
\end{aligned}
\end{align}

By comparing \eqref{eqn::hh} and \eqref{eqn::hh_upsilon}, to prove \eqref{eqn::generic_alpha_suff_cond}, it is sufficient to show 
\begin{align}
\label{eqn::generic_alpha_suff_cond_v2}
\Inner{\wt{A}}{\Upsilon^{\kappa} (\wt{A})}_{\Omega_{\sigma}^{\kappa}} 
\le \eta_{\max}^{\kappa} \Inner{\wt{A}}{\wt{A}}_{\Omega_{\sigma}^{\kappa}}.
\end{align}
Thus we complete the proof of \lemref{lem::tilde_H}.

\end{proof}

\section{Connection to the quantum maximal correlation}
\label{sec::max_corr}

The SDPI constant for the classical $\chi^2$ divergence is closely connected to the classical maximal correlation (see \eg, \cite[Theorem \rom{3}.2]{Raginsky_strong_2016}). In the proposition below, we provide a quantum analog of this relation when $\kappa = \kappa_{1/2}$. 

To begin with, we need to define the quantum maximal correlation. This concept was previously proposed and studied in \cite{beigi_new_2013}. Since there is a whole family of quantum $\chisq$ divergences, it is natural to imagine that there could also exist a whole family of quantum maximal correlations, as a straightforward generalization of \cite{beigi_new_2013}.

\begin{defn}[$\kappa$-quantum maximal correlation]
\label{defn::max_corr}
Consider any fixed $\kappa \in \mathcal{K}$ and Hilbert spaces $\hbt_1$ and $\hbt_2$ with dimensions $n_1$ and $n_2$ respectively. 
For any bipartite quantum state $\rho_{1,2}$ on the composite system $\hbt_1\otimes \hbt_2$, 
denote the reduced density matrices by $\rho_1$ and $\rho_2$ respectively (\ie, $\tr_{2}(\rho_{1,2}) = \rho_1$, $\tr_{1}(\rho_{1,2}) = \rho_2$). 
Define the $\kappa$-quantum maximal correlation $\mu_{\kappa}(\rho_{1,2})$  by
\begin{align}
\mu_{\kappa}(\rho_{1,2}) := \max \Abs{\tr\left(\rho_{1,2} F \otimes G^{\dagger}\right)},
\end{align}
where the maximum is taken over all $F\in \MM_{n_1}$, $G\in \MM_{n_2}$ such that 
\begin{align}
\label{eqn::max_corr_condition}
\Inner{\unit_{n_1}}{F}_{\mho_{\rho_1}^{\kappa}} = \Inner{\unit_{n_2}}{G}_{\mho_{\rho_2}^{\kappa}} = 0,
\qquad
\Norm{F}_{\mho_{\rho_1}^{\kappa}} = \Norm{G}_{\mho_{\rho_2}^{\kappa}}  = 1. 
\end{align}
\end{defn}

\revise{Technically, when $\rho_1$ is not a full-rank density matrix, the notation $\Inner{\cdot}{\cdot}_{\mho_{\rho_1}^{\kappa}}$ should be understood as a sesquilinear form, as we explained at the beginning of \secref{sec::preliminary} and the operator $\mho_{\rho_1}^{\kappa}$ is still well-defined on the support of $\rho_1$ via \eqref{eqn::mho_decomp}.
}
By \lemref{lem::mho_op}, we easily verify that $\Inner{\unit_{n_1}}{F}_{\mho_{\rho_1}^{\kappa}} \equiv \tr(\rho_1 F)$ and $\Inner{\unit_{n_2}}{G}_{\mho_{\rho_2}^{\kappa}} \equiv \tr(\rho_2 G)$. When $\kappa(x) = 1$ is a constant function, we recover the quantum maximal correlation defined in \cite{beigi_new_2013}; in this case, $\mho_{\sigma}^{\kappa(x) = 1} = L_{\sigma}$;
however, notice that this choice of $\kappa$ is not included in the set $\mathcal{K}$ and the corresponding operator $\mho_{\sigma}^{\kappa(x)=1}$ is not Hermitian-preserving. 

\begin{lemma}[Invariance of the $\kappa$-quantum maximal correlation under local isometries]
\label{lem::max_corr_isometry}
Suppose $U: \hbt_{1} \rightarrow \wt{\hbt}_{1}$ and $V: \hbt_{2}\rightarrow \wt{\hbt}_{2}$ are two isometries (\ie, $U^{\dagger} U = \unit_{{\dim(\hbt_1)}}$ and $V^{\dagger} V = \unit_{{\dim(\hbt_2)}}$),
where $\dim(\hbt_{1}) \le \dim(\wt{\hbt}_{1})$ and $\dim(\hbt_{2}) \le \dim(\wt{\hbt}_{2})$. 
For any bipartite quantum state $\rho_{1,2}$ on $\hbt_{1} \otimes \hbt_{2}$, define $\wt{\rho}_{1,2} := (U \otimes V) \rho (U \otimes V)^{\dagger}$. We have
\begin{align}
\mu_{\kappa}(\rho_{1,2}) = \mu_{\kappa}(\wt{\rho}_{1,2}). 
\end{align}
\end{lemma}

\begin{proof}
By definition, 
\begin{align*}
\mu_{\kappa}(\wt{\rho}_{1,2}) &= \max_{\wt{F}, \wt{G}} \Abs{\tr\left(\wt{\rho}_{1,2}\ \wt{F} \otimes \wt{G}^{\dagger}\right)}  \\
&= \max_{\wt{F}, \wt{G}} \Abs{\tr\left(\rho_{1,2}\ (U^{\dagger} \wt{F} U) \otimes (V^{\dagger} \wt{G}^{\dagger} V) \right)} = \max_{\wt{F}, \wt{G}} \Abs{\tr\left(\rho_{1,2} F \otimes G^{\dagger} \right)},
\end{align*}
where we define $F := U^{\dagger} \wt{F} U$ and $G:= V^{\dagger} \wt{G} V$.
Denote the reduced density matrices of $\rho_{1,2}$ as $\rho_1$ and $\rho_2$ respectively. Then the reduced density matrices of $\wt{\rho}_{1,2}$ are given by $\wt{\rho}_1 := U \rho_{1} U^{\dagger}$ and $\wt{\rho}_2 := V \rho_{2} V^{\dagger}$ respectively. 
From \eqref{eqn::max_corr_condition}, the condition in the maximization is given by
\begin{align*}
\tr\left(\wt{\rho}_{1} \wt{F} \right) =  \tr\left(\wt{\rho}_2 \wt{G}\right) = 0, \qquad \tr\left( \wt{F}^{\dagger} \mho_{\wt{\rho}_1}^{\kappa} (\wt{F})  \right) = \tr\left( \wt{G}^{\dagger} \mho_{\wt{\rho}_{2}}^{\kappa} (\wt{G}) \right) = 1.
\end{align*}
By \eqref{eqn::mho_decomp}, it could be readily shown that $\mho_{\wt{\rho}_{1}}^{\kappa}(\cdot) = (U\# U^{\dagger}) \circ \mho_{\rho_{1}}^{\kappa} \circ \left( U^{\dagger}\# U \right)$ and similarly for $\mho_{\wt{\rho}_2}^{\kappa}(\cdot)$.  
As a remark, in this case, $\wt{\rho}_1$ and $\wt{\rho}_2$ might not be strictly positive, then the decomposition in \eqref{eqn::mho_decomp} only considers eigenstates with respect to non-zero eigenvalues (\ie, $\mho_{\wt{\rho}_1}^{\kappa}$ is only defined on the support of $\wt{\rho}_1$).
Then, with direct calculation, one could verify that the above four conditions are equivalent to 
\begin{align*}
\tr\left(\rho_1 F \right) = \tr\left(\rho_2 G \right) = 0, \qquad 
\tr\left(F^{\dagger} \mho_{\rho_{1}}^{\kappa} (F)\right) = \tr\left(G^{\dagger} \mho_{\rho_2}^{\kappa}(G)\right) = 1. 
\end{align*}
Therefore, we know $\mu_{\kappa}(\wt{\rho}_{1,2}) \le \mu_{\kappa}(\rho_{1,2})$. Since $\wt{F}$ is a linear operator on a higher-dimensional Hilbert space $\wt{\hbt}_{1}$ than $F$ on $\hbt_{1}$, for any such $F$, there exists $\wt{F}$ such that $U^{\dagger} \wt{F} U = F$ (similarly for $G$); therefore the equality can be achieved and $\mu_{\kappa}(\wt{\rho}_{1,2}) = \mu_{\kappa}(\rho_{1,2})$.
\end{proof}

\begin{theorem}
\label{thm::sdpi_max_corr}
For a Hilbert space $\hbt$ with dimension $n$, 
suppose $\sigma\in \dsetplus_{n}$ and $\mathcal{E}$ is any  quantum channel on $\hbt$ such that the quantum state $\mathcal{E}(\sigma)\in \dsetplus_n$.
Thus, $\sigma$ has an eigenvalue decomposition 
$
\sigma = \sum_{j=1}^{n} s_j \Ket{s_j}\Bra{s_j}.
$
For the choice $\kappa = \kappa_{1/2}$, 
\begin{align}
\sqrt{\eta_{\chi^2_{{\kappa_{1/2}}}}(\mathcal{E}, \sigma)} = \mu_{\kappa_{1/2}}(\rho_{1,2}),
\end{align}
where the bipartite quantum state 
$\rho_{1,2} := (\unit_{n} \otimes \mathcal{E}) \left(\Ket{\psi} \Bra{\psi}\right)$ 
and the wave function $\Ket{\psi}$ is any purification of $\sigma$ on the system $\hbt \otimes \hbt$.
\end{theorem}

Recall that a pure state $\Ket{\psi}$ on $\hbt\otimes \hbt$ is a purification of $\sigma$ if $\tr_{1}\left(\Ket{\psi}\Bra{\psi}\right) = \sigma$ (see \cite[Chap. 5]{wilde_2013}). The canonical choice of the purification $\Ket{\psi}$ of $\sigma$ is  
\begin{align}
\label{eqn::sigma_purification}
\Ket{\psi_{\text{c}}} := \sum_{j=1}^{n} \sqrt{s_j} \Ket{s_j, s_j}. 
\end{align}

\begin{proof}
In the first step, we prove it for the choice $\Ket{\psi} = \Ket{\psi_c}$; in the second step, we extend the result to the general purification. 

{\noindent \emph{Step (\rom{1})}.}
By \lemref{lem::sdpi_svd}, we have
\begin{align*}
\sqrt{\eta_{\chisq}(\mathcal{E}, \sigma)} = \max_{F, G} \Abs{\Inner{\mathcal{E} \circ \Gamma_{\sigma}(F)}{\wt{G}}_{\hs} } = \max_{F, G} \Abs{\Inner{\wt{G}}{\mathcal{E} \circ \Gamma_{\sigma}(F)}_{\hs} }, 
\end{align*}
where $\wt{G} := \left(\Gamma_{\mathcal{E}(\sigma)}\right)^{-1}\circ \mho_{\mathcal{E}(\sigma)}^{\kappa} (G)$. 
Let us decompose $\mathcal{E} \circ \Gamma_{\sigma}(F)$ based on the eigenstates of $\sigma$, 
\begin{align*}
\mathcal{E} \circ \Gamma_{\sigma}(F) = \sum_{j,m=1}^{n} \sqrt{ s_m s_j} \Bra{s_m} F \Ket{s_j} \mathcal{E}(\Ket{s_m} \Bra{s_j}).
\end{align*}
Hence, 
\begin{align*}
\sqrt{\eta_{\chisq}(\mathcal{E}, \sigma)} &= \max_{F, G} \Abs{\sum_{j,m=1}^{n} \sqrt{s_m s_j} \Bra{s_m} F \Ket{s_j} \Inner{\wt{G}}{\mathcal{E}(\Ket{s_m} \Bra{s_j})}_{\hs}} \\
&= \max_{F, G} \Abs{\sum_{j,m=1}^{n} \sqrt{s_m s_j} \Bra{s_j} \wt{F} \Ket{s_m} \Inner{\wt{G}}{\mathcal{E}(\Ket{s_m} \Bra{s_j})}_{\hs}}\\
&= \max_{F, G} \Abs{\tr\left(\rho_{1,2} \wt{F}\otimes \wt{G}^{\dagger} \right)},
\end{align*}
where $\wt{F} = F^{T}$ and the superscript $T$ means transpose with respect to the eigenstates of $\sigma$, \ie, $\Bra{s_j} \wt{F} \Ket{s_m} := \Bra{s_m} F \Ket{s_j}$ for all $1\le j, m \le n$. The last equality above can be verified directly by $\rho_{1,2} = (\id_{n} \otimes \mathcal{E}) (\Ket{\psi_c}\Bra{\psi_c})$.

Notice that from \lemref{lem::sdpi_svd}, the maximum is taken over all $F, G$ given in \eqref{eqn::sdpi_max_corr_cond}. 
Hence, to prove \thmref{thm::sdpi_max_corr}, it remains to verify that conditions \eqref{eqn::sdpi_max_corr_cond} for $F$ and $G$ are equivalent to conditions \eqref{eqn::max_corr_condition} for $\wt{F}$ and $\wt{G}$. 
More specifically, we need to verify the following four relations.
\begin{enumerate}[(i)]
\item $\Inner{\unit_n}{F}_{\mho_{\sigma}^{\kappa}} = \Inner{\unit_n}{\wt{F}}_{\mho_{\sigma}^{\kappa}}$. Note that 
\begin{align*}
\Inner{\unit_n}{F}_{\mho_{\sigma}^{\kappa}} = \tr(\sigma F)  = \sum_{j} s_j \Bra{s_j} F \Ket{s_j} = \sum_{j} s_j \Bra{s_j} \wt{F} \Ket{s_j} = \tr(\sigma \wt{F}) = \Inner{\unit_n}{\wt{F}}_{\mho_{\sigma}^{\kappa}}.
\end{align*}

\item $\Norm{F}_{\mho_{\sigma}^{\kappa}} = \Norm{\wt{F}}_{\mho_{\sigma}^{\kappa}}$. Note that 
\begin{align*}
\Norm{F}_{\mho_{\sigma}^{\kappa}}^2  &= \tr\left(F^\dagger \mho_{\sigma}^{\kappa} (F)\right) 
\myeq{\eqref{eqn::mho_decomp}}\ \sum_{j,m} s_j \kappa\left(\frac{s_j}{s_m}\right) \Abs{\Bra{s_j} F \Ket{s_m} }^2\\
&= \sum_{j,m} s_j \kappa\left(\frac{s_j}{s_m}\right) \Abs{\Bra{s_m} \wt{F} \Ket{s_j} }^2 
\myeq{\eqref{eqn::set_K}} \sum_{j,m} s_m \kappa\left(\frac{s_m}{s_j}\right) \Abs{\Bra{s_m} \wt{F} \Ket{s_j} }^2 \\
&= \tr\left(\bigl(\wt{F}\bigr)^\dagger \mho_{\sigma}^{\kappa}(\wt{F})\right) = \Norm{\wt{F}}^2_{\mho_{\sigma}^{\kappa}}.
\end{align*}

\item $\Inner{\unit_n}{G}_{\mho_{\mathcal{E}(\sigma)}^{\kappa}} = \Inner{\unit_n}{\wt{G}}_{\mho_{\mathcal{E}(\sigma)}^{\kappa}}$. Note that
\begin{align*}
\Inner{\unit_n}{\wt{G}}_{\mho_{\mathcal{E}(\sigma)}^{\kappa}} &= \Inner{\mathcal{E}(\sigma)}{ \Gamma_{\mathcal{E}(\sigma)}^{-1}\circ \mho_{\mathcal{E}(\sigma)}^{\kappa} (G)}_{\hs} = \Inner{\mathcal{E}(\sigma)}{ G}_{\hs} = \Inner{\unit_n}{G}_{\mho_{\mathcal{E}(\sigma)}^{\kappa}}.
\end{align*}

\item $\Norm{G}_{\mho_{\mathcal{E}(\sigma)}^{\kappa}} = \Norm{\wt{G}}_{\mho_{\mathcal{E}(\sigma)}^{\kappa}}$.
Note that 
\begin{align*}
\Norm{\wt{G}}_{\mho_{\mathcal{E}(\sigma)}^{\kappa}}^2 
&= \Inner{\wt{G}}{\mho_{\mathcal{E}(\sigma)}^{\kappa}\left(\wt{G}\right)}_{\hs} 
= \Inner{\Gamma_{\mathcal{E}(\sigma)}^{-1}\circ \mho_{\mathcal{E}(\sigma)}^{\kappa} (G)}{\mho_{\mathcal{E}(\sigma)}^{\kappa} \circ \Gamma_{\mathcal{E}(\sigma)}^{-1}\circ \mho_{\mathcal{E}(\sigma)}^{\kappa} (G)}_{\hs} \\
&= \Inner{G}{\left(\Gamma_{\mathcal{E}(\sigma)}^{-1}\circ \mho_{\mathcal{E}(\sigma)}^{\kappa} \right)^2 (G)}_{\mho_{\mathcal{E}(\sigma)}^{\kappa}}. 
\end{align*}
When $\kappa = \kappa_{1/2}$, $\Gamma_{\mathcal{E}(\sigma)}^{-1}\circ \mho_{\mathcal{E}(\sigma)}^{\kappa} = \id_n$. Thus the relation holds for this special choice of $\kappa$ and this is the only place we employ this assumption.
\end{enumerate}

{\noindent \emph{Step (\rom{2}):}} We then extend the result from the canonical purification $\Ket{\psi_c}$ to any purification $\Ket{\psi}$ on the bipartite quantum system $\hbt \otimes \hbt$. 
By \cite[Theorem 5.1.1]{wilde_2013}, there exists a unitary (thus also isometry) $U: \hbt \rightarrow \hbt$ such that $\Ket{\psi} = U \otimes \unit_n \Ket{\psi_c}$. 
Hence, $(\id_n \otimes \mathcal{E})(\Ket{\psi}\Bra{\psi}) = (U\otimes \unit_n) \left((\id_n\otimes \mathcal{E})(\Ket{\psi_c}\Bra{\psi_c})\right) (U\otimes \unit_n)^{\dagger}$. 
By \lemref{lem::max_corr_isometry}, the conclusion follows immediately. 
\end{proof}

\section{SDPI constants for special {qubit channels}}
\label{sec::example}

In this section, we will illustrate the dependence of SDPI constants on the reference state $\sigma$ and the weight function $\kappa$, for several special \revise{qubit channels}. The dependence on $\sigma$ is one major difference between the quantum SDPI framework and the quantum contraction coefficient approach. The dependence on $\kappa$ is one major difference between the quantum SDPI framework and its classical version: 
\revise{all quantum $\chisq$ divergences coincide for classical states $\rho$ and $\sigma$ (\ie, $\rho$ and $\sigma$ commute) and simply reduce to the classical $\chi^2$ divergence; in particular, 
classical $\chi^2$ divergence, as well as the associated classical SDPI constant, does not depend on $\kappa$; however, the SDPI constant for quantum $\chisq$ divergences might fluctuate significantly between approximately $0$ and $1$ for various $\kappa$, in a special example that we provide below.  
}

Three Pauli matrices are denoted by $\sigma_{X}, \sigma_{Y}, \sigma_{Z}$. 
Without loss of generality, assume $\sigma = \frac{1}{2}\left(\unit_2 + s \sigma_{Z}\right) = \bigl[\begin{smallmatrix}
(1+s)/2 & 0 \\ 0 & (1-s)/2
\end{smallmatrix}\bigr]$ with $s\in [0,1)$, because one can always choose the eigenbasis of $\sigma$ as the computational basis; of course, the matrix representation of the quantum channel is changed, by choosing such a specific computational basis. 

\subsection{QC channel}
\label{subsec::qc_channel}

By the expression of QC channel \eqref{eqn::qc_channel} and by \eqref{eqn::qc_value}, 
we have for any $A\in \HH_2^0$ that 
\begin{align*}
\Inner{\mathcal{E}(A)}{\Omega_{\mathcal{E}(\sigma)}^{\kappa} (\mathcal{E}(A))}_{\hs} = \sum_{j=1}^{2} \frac{\Abs{\tr(F_j A)}^2}{\tr(F_j \sigma)} = \Abs{\tr(F_1 A)}^2 \left( \frac{1}{\tr(F_1 \sigma)} + \frac{1}{1 - \tr(F_1 \sigma)} \right).
\end{align*}
The second equality comes from the fact that $F_2 = \unit_2 -  F_1$ and $\tr(A) = 0$. Let us decompose $A = a_x \sigma_{X} + a_y \sigma_{Y} + a_z \sigma_{Z}$ and $F_1 = f_0 \unit_2 + f_x \sigma_{X} + f_y \sigma_{Y} + f_{z} \sigma_{Z}$; notice that all coefficients for $A$ and $F_1$ are real numbers. Next, rewrite the above equation by 
\begin{align*}
\Inner{\mathcal{E}(A)}{\Omega_{\mathcal{E}(\sigma)}^{\kappa} (\mathcal{E}(A))}_{\hs} =  \frac{4}{\tr(F_1 \sigma) \bigl(1 - \tr(F_1 \sigma)\bigr)}  \left(f_x a_x + f_y a_y + f_z a_z\right)^2. 
\end{align*}
From \eqref{eqn::AA_omega_kappa}, we also have
\begin{align*}
\Inner{A}{\Omega_{\sigma}^{\kappa} (A)}_{\hs} =  c_s (a_x^2 + a_y^2) + \frac{4}{1-s^2} a_z^2,
\end{align*}
where 
\begin{align}
\label{eqn::qc_channel_c_s}
c_s := \kappa\left(\frac{1+s}{1-s}\right)\frac{2}{1-s} + \kappa\left(\frac{1-s}{1+s}\right)\frac{2}{1+s} \myeq{\eqref{eqn::set_K}} \frac{4}{1-s} \kappa\left(\frac{1+s}{1-s}\right) \myeq{\eqref{eqn::set_K}} \frac{4}{1+s} \kappa\left(\frac{1-s}{1+s}\right). 
\end{align}
By the Cauchy–Schwarz inequality and the fact that $1-s^2>0$ and $c_s > 0$, we have 
\begin{align*}
\left(f_x a_x + f_y a_y + f_z a_z\right)^2 \le \left(\frac{f_x^2}{c_s} + \frac{f_y^2}{c_s} + \frac{f_z^2}{4/(1-s^2)}\right)  \left(c_s (a_x^2 + a_y^2) + \frac{4}{1-s^2} a_z^2\right).
\end{align*}
Hence, we know that 
\begin{align}
\label{eqn::sdpi_const_qc}
\begin{split}
\eta_{\chisq}(\mathcal{E}, \sigma) &= \frac{4}{\tr(F_1 \sigma) \bigl(1 - \tr(F_1 \sigma)\bigr)}\left(\frac{f_x^2}{c_s} + \frac{f_y^2}{c_s} + \frac{f_z^2}{4/(1-s^2)}\right) \\
&= \frac{4}{(f_0 + s f_z)(1 - f_0 - s f_z)} \left(\frac{f_x^2}{c_s} + \frac{f_y^2}{c_s} + \frac{f_z^2}{4/(1-s^2)}\right).
\end{split}
\end{align}
As we can observe, the SDPI constant $\eta_{\chisq}(\mathcal{E}, \sigma)$ depends on $\kappa$ and the parameter $s$ in a complicated way; however, it does not depend on the choice of pure states in the post-measurement preparation in \eqref{eqn::qc_channel}. 
In the following, let us consider a few special choices of the POVM $\{F_1, \unit_2 - F_1\}$.

\subsubsection*{(High dependence on $\sigma$, for the quantum implementation of BSC)}
If $F_1 = \begin{bmatrix}1-\eps & 0 \\ 0 & \eps \end{bmatrix}$ (thus $F_2 = \begin{bmatrix} \eps & 0 \\ 0 & 1-\eps \end{bmatrix}$) with $\eps\in[0,1]$, 
then the channel $\mathcal{E}$ is exactly a quantum implementation of the binary symmetric channel with crossover probability $\eps$ (or BSC($\eps$) in short). 
Easily, we know $f_0 = \frac{1}{2}$, $f_z = \frac{1-2\eps}{2}$, $f_x = f_y = 0$ and thus the SDPI constant can be simplified as 
\begin{align}
\label{eqn::sdpi_constant_bsc_eps}
\eta_{\chisq}(\mathcal{E}, \sigma) 
= (1-2\eps)^2 \frac{1-s^2}{1-(1-2\eps)^2 s^2} \le (1-2\eps)^2. 
\end{align}
Notice that the SDPI constant in this case is independent of the choice of $\kappa$; the upper bound comes from the fact that $\eps \in [0,1]$.  
When we further let $s = 0$, \ie, the reference state $\sigma$ has the distribution Bern($\frac{1}{2}$), the SDPI constant achieves the upper bound $(1 - 2\eps)^2$, which recovers \cite[Example \rom{3}.1]{Raginsky_strong_2016}. 
In \figref{fig::bsc}, we show $\eta_{\chisq}(\mathcal{E}, \sigma)$ with respect to the parameter $s$ in $\sigma$, for fixed $\eps = 0.05$; the high dependence of $\eta_{\chisq}$ on $s$ (\ie, on $\sigma$) can be clearly seen, for this particular case.

\subsubsection*{(High dependence on $\kappa$).}

If $F_1 = \frac{1}{2}\left(\unit_2 + \xi \sigma_X\right)$ with $\xi\in [-1,1]$, then $f_0 = \frac{1}{2}$, $f_x = \frac{\xi}{2}$ and $f_{y} = f_z = 0$. Hence, 
\begin{align}
\label{eqn::sdpi_constant_qc_example_two_1}
\eta_{\chisq}(\mathcal{E}, \sigma) = \frac{4\xi^2}{c_s} \myeq{\eqref{eqn::qc_channel_c_s}} \frac{(1+s)\xi^2}{\kappa(\frac{1-s}{1+s})}\le \xi^2.
\end{align}
The inequality comes from the fact for any $\kappa\in \mathcal{K}$, we have $\kappa(x) \ge \kappa_{\min}(x) \equiv \frac{2}{1+x}$ (see \cite[Eq. (11)]{hiai_contraction_2016}). 
As one could observe, even for this simple example,  
the dependence of $\eta_{\chisq}(\mathcal{E}, \sigma)$ on $s$ and $\kappa$ is nonlinear and slightly complicated.
Similarly,  by the fact that for any $\kappa\in \mathcal{K}$, we have $\kappa(x) \le \kappa_{\max}(x) \equiv \frac{1+x}{2x}$ (see \cite[Eq. (11)]{hiai_contraction_2016}), one could immediately show that 
\begin{align}
\label{eqn::sdpi_constant_qc_example_two_1}
\eta_{\chisq}(\mathcal{E}, \sigma) 
\ge (1-s^2)\xi^2.
\end{align}
Notice that both upper and lower bounds in the above can be achieved for some $\kappa\in \mathcal{K}$. 
When $s\approx 1$ and $\xi\approx 1$, the largest value of $\eta_{\chisq}(\mathcal{E}, \sigma)$ is approximately $1$, while the smallest value is approximately $0$, which illustrates the high dependence of $\eta_{\chisq}(\mathcal{E}, \sigma)$ on the choice of $\kappa$, for this extreme case. In \figref{fig::qc_wyd_chi_alpha}, we visualize the SDPI constant $\eta_{\chisq}(\mathcal{E}, \sigma)$ with respect to various choices of $\kappa$, for $\xi = s = 0.95$; the high dependence of $\eta_{\chisq}(\mathcal{E}, \sigma)$ on $\kappa$ can be clearly observed.

\begin{figure}
\centering
\includegraphics[width=0.45\textwidth]{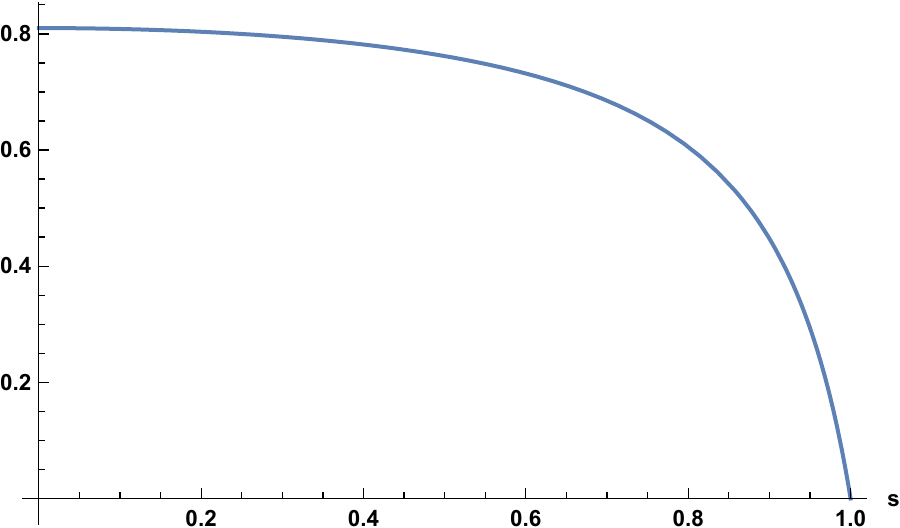}
\caption{The SDPI constant $\eta_{\chisq}(\mathcal{E}, \sigma)$ with respect to $s$, for BSC(0.05).}
\label{fig::bsc}
\end{figure}

\begin{figure}
\centering
\begin{subfigure}[b]{0.45\textwidth}
\includegraphics[width=\textwidth]{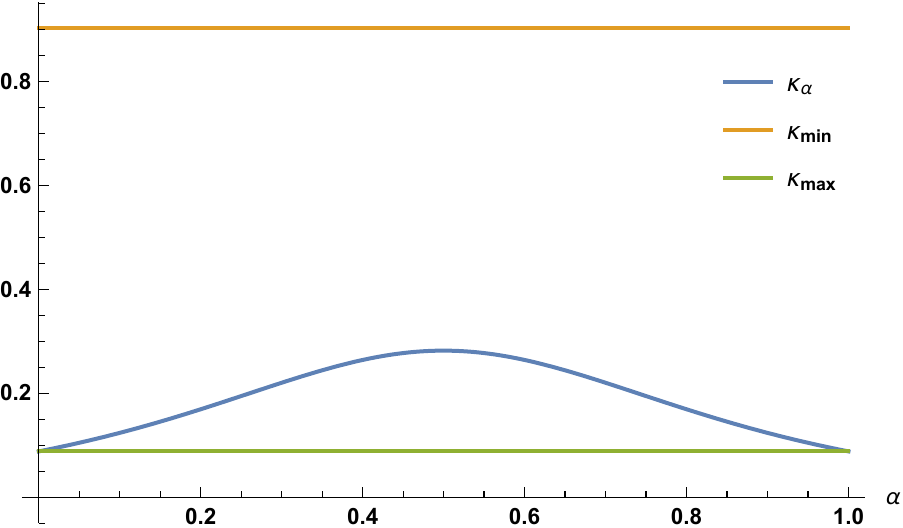}
\caption{$\eta_{\chisq}(\mathcal{E}, \sigma)$ for the family of $\kappa_{\alpha}$ in \eqref{eqn::kappa-chi-alpha}.}
\end{subfigure}
~
\begin{subfigure}[b]{0.45\textwidth}
\includegraphics[width=\textwidth]{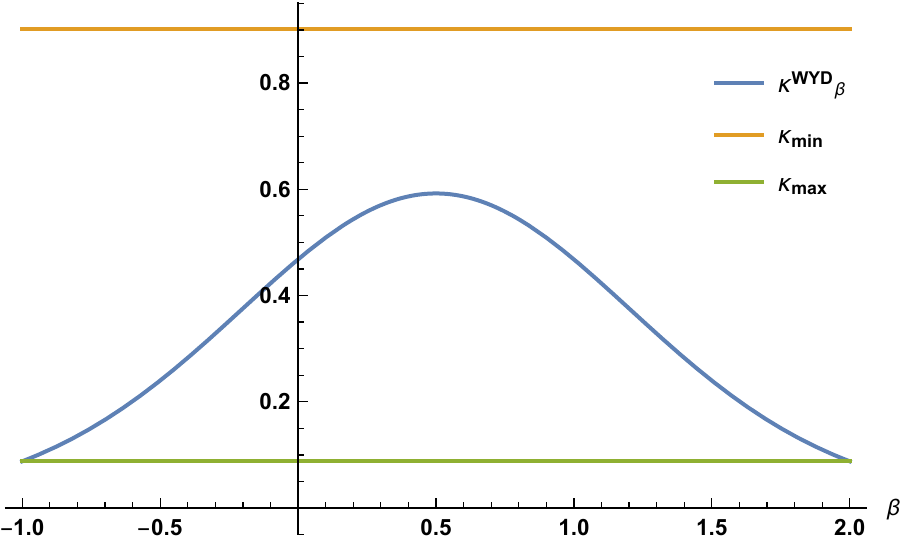}
\caption{$\eta_{\chisq}(\mathcal{E}, \sigma)$ for the family of $\kappa_{\beta}^{WYD}$ in \eqref{eqn::kappa_WYD}. }
\end{subfigure}
\caption{The SDPI constant $\eta_{\chisq}(\mathcal{E}, \sigma)$ with respect to various choices of $\kappa$ for the QC channel $\mathcal{E}$ with $F_1 = \frac{1}{2}\left(\unit_2 + \xi \sigma_X\right)$, $\xi = 0.95$, and $s = 0.95$.}
\label{fig::qc_wyd_chi_alpha}
\end{figure}

\subsection{Depolarizing channel}
The depolarizing channel on a qubit has the following form
\begin{align*}
\mathcal{E}(\rho) = \eps \rho + (1-\eps) \tr(\rho) \frac{\unit_2}{2},
\end{align*}
for $\eps\in [0,1]$. It refers to a physical process in which for a given input state $\rho$, one prepares $\rho$ with probability $\eps$ and prepares the maximal mixed state $\frac{\unit_2}{2}$ with probability $1-\eps$. Easily, we know that $\mathcal{E}(\sigma) = \frac{\unit_2}{2} + \frac{s \eps}{2}\sigma_{Z}$, and $\mathcal{E}(A) = \eps A$ for any $A\in \HH_2^0$. Hence,
\begin{align*}
\frac{\Inner{\mathcal{E}(A)}{\Omega_{\mathcal{E}(\sigma)}^{\kappa} (\mathcal{E}(A))}_{\hs}}{\Inner{A}{\Omega_{\sigma}^{\kappa} (A)}_{\hs} } &= \eps^2\frac{\Inner{A}{\Omega_{\mathcal{E}(\sigma)}^{\kappa}(A)}_{\hs}}{\Inner{A}{\Omega_{\sigma}^{\kappa} (A)}_{\hs}} \myeq{\eqref{eqn::Omega_decomp}} \eps^2 \frac{\frac{4}{1-s^2\eps^2} a_z^2 + c_{s\eps} (a_x^2 + a_y^2)}{\frac{4}{1-s^2} a_z^2 + c_s (a_x^2 + a_y^2)} \\
&= \eps^2 \left(\frac{1-s^2}{1-s^2 \eps^2} + \frac{c_{s\eps} - \frac{1-s^2}{1-s^2\eps^2} c_s}{\frac{4}{1-s^2} \frac{a_z^2}{a_x^2 + a_y^2} + c_s}\right).
\end{align*}

If $c_{s\eps} - \frac{1-s^2}{1-s^2 \eps^2} c_s \ge 0$, then 
\begin{align*}
\eta_{\chisq}(\mathcal{E}, \sigma) = \eps^2 \frac{c_{s\eps}}{c_s} =\eps^2 \frac{\kappa\left(\frac{1+s\eps}{1-s\eps}\right)/(1-s\eps)}{ \kappa\left(\frac{1+s}{1-s}\right)/(1-s)}.
\end{align*}
For fixed $s$ and $\eps$, $\eta_{\chisq}(\mathcal{E}, \sigma)$ might be largely affected by $\kappa$ as well.

\section{Conclusion and outlook}
\label{sec::discussion}

In this paper, we provide a partial solution to the problem of the
tensorization of SDPIs for quantum channels in
\thmref{thm::chisq_tensorize}. In addition, we extend the connection
between the SDPI constant for classical $\chi^2$ divergence and the maximal
correlation to the quantum region in \thmref{thm::sdpi_max_corr}. For
a particular QC channel $\mathcal{E}$ and a special quantum state
$\sigma$, we observe an extreme scenario, in which the SDPI constant
$\eta_{\chisq}(\mathcal{E}, \sigma)$ ranges approximately from $0$ to
$1$ for different $\kappa\in \mathcal{K}$. This implies that choosing
different $\kappa$ might largely affect the rate of contraction of
quantum channels.  \revise{Our numerical experiments (not presented in
  the paper) conducted for both qubit (\ie, $n = 2$) and qudit (with
  $n=3$) systems show that the tensorization property
  \eqref{eq:mainthm} seems to hold for any quantum channel
  $\mathcal{E}$, any reference state $\sigma\in \dsetplus_n$ and at
  least a few weight functions $\kappa$ being tested (\eg{},
  $\kappa_{\min} \equiv \frac{2}{1+x}$,
  $\kappa_{\max} \equiv \frac{1+x}{2x}$ and the family
  $\kappa_{\alpha}$ with $\alpha=\frac{1}{4}$ and
  $\alpha=\frac{3}{4}$). Proving such tensorization properties is an
  interesting future work. }

\revise{Finally, let us comment on the potential generalization of our
  approach, as well as the limitation.  As one might observe, provided
  that one could show \eqref{eqn::monotonicity_sdpiratio}, the
  tensorization of SDPIs is an immediate consequence.  However, it
  seems to be challenging to characterize the class of quantum
  channels that satisfy \eqref{eqn::monotonicity_sdpiratio} in general
  and this is the reason why we restrict to QC channels and the case
  $\kappa \ge \kappa_{1/2}$ in \thmref{thm::chisq_tensorize}.  In
  terms of the validity of \eqref{eqn::monotonicity_sdpiratio}, we
  notice that when $\kappa \le \kappa_{1/2}$ (\eg, $\kappa_{\min}$),
  \eqref{eqn::monotonicity_sdpiratio} does not hold even for QC
  channels.  As mentioned above, numerical experiments seem to suggest
  that the tensorization also holds for $\kappa_{\min}$.  Therefore,
  further understanding of the properties of quantum $\chi^2_{\kappa}$
  divergences is needed to extend our results.  }

\section*{Acknowledgment}

This work is supported in part by the US National Science Foundation
via grants DMS-1454939 and CCF-1910571 and by the US Department of
Energy via grant {DE}-{SC}0019449.  We thank Iman Marvian and Henry
Pfister for helpful discussions. Iman Marvian pointed out the possible
generalization of \thmref{thm::sdpi_max_corr} from the canonical
purification to any general purification. Henry Pfister introduced us
to the topic of the strong data processing inequality for classical noisy
channels.  \revise{We also thank anonymous referees for helpful suggestions.}

\bibliographystyle{plainnat}
\bibliography{reference.bib}

\end{document}